\documentclass[conference]{IEEEtran}
\IEEEoverridecommandlockouts

\usepackage[author=anonymous,nomargin,marginclue,footnote,final]{fixme}
\FXRegisterAuthor{ls}{als}{LS}
\FXRegisterAuthor{pw}{apw}{PW}

\usepackage{amsmath}
\usepackage{amssymb}
\usepackage{amsfonts}

\usepackage{amsthm}

\usepackage{mathtools}
\usepackage{stmaryrd}

\usepackage[T1]{fontenc}
\usepackage[utf8]{inputenc}

\usepackage{tikz}
\usetikzlibrary{cd}

\usepackage[colorlinks,linkcolor=blue]{hyperref}

\usepackage{trimclip} 

\usepackage{enumitem} 

\usepackage{mathabx} 

\newcommand*{\V}{\mathcal{V}}

\newcommand*{\vtimes}{\otimes}
\newcommand*{\vminel}{\bot}
\newcommand*{\vmaxel}{\top}
\newcommand*{\vunit}{k}

\newcommand*{\vJoin}{\bigsqcup}
\newcommand*{\vMeet}{\bigsqcap}

\newcommand*{\vle}{\sqsubseteq}
\newcommand*{\vge}{\sqsupseteq}

\newcommand*{\twoQ}{{\mathsf{2}_{\land}}}

\newcommand*{\zeroinfQ}{{[0,\infty]_{+}}}

\newcommand*{\id}{\mathsf{id}}

\newcommand*{\supp}{\mathsf{supp}}

\newcommand*{\Cpl}{\mathsf{Cpl}}

\newcommand*{\Kant}{\mathsf{Kant}}
\newcommand*{\Ham}{\mathsf{Ham}}
\newcommand*{\Lev}{\mathsf{Lev}}

\newcommand*{\List}{\mathsf{List}}
\newcommand*{\Wc}{W^{=}}
\newcommand*{\bd}{\mathsf{bd}}

\newcommand*{\Set}{\mathsf{Set}}
\newcommand*{\PreOrd}{\mathsf{PreOrd}}
\newcommand*{\HMet}{\mathsf{HMet}}

\newcommand*{\op}{\mathsf{op}}

\newcommand*{\VRel}{\V\text{-}\mathsf{Rel}}
\newcommand*{\VCat}{\V\text{-}\mathsf{Cat}}
\newcommand*{\VDist}{\V\text{-}\mathsf{Dist}}

\newcommand*{\CA}{\mathcal{A}}

\newcommand*{\dfun}{\mathcal{D}_\omega}

\newcommand*{\pfun}{\mathcal{P}}
\newcommand*{\Pow}{\pfun}
\newcommand*{\Bag}{\mathcal{B}}

\newcommand*{\Diamonda}[1]{\langle{#1}\rangle}

\newcommand*{\frel}[2]{#1 \mathbin{\ooalign{$\rightarrow$\cr$\hspace{0.12ex}+$\cr}} #2}

\newcommand*{\rev}[1]{{#1}^\circ}
\newcommand*{\gr}[1]{{#1}_\circ}

\newcommand*{\diag}[1]{\mathsf{dia}(#1)}
\newcommand*{\lambdadiag}{\lambda^\mathsf{dia}}

\newcommand*{\Real}{\mathbb{R}}

\newcommand*{\Nat}{\mathbb{N}}

\newcommand*{\cpl}[2]{\mathsf{Cpl} (#1,#2)}

\newcommand*{\expect}[1]{\mathbb{E}_{#1}}

\newcommand*{\A}{\mathcal{A}}

\newcommand*{\leready}{\le_\mathsf{rd}}
\newcommand*{\geready}{\ge_\mathsf{rd}}
\newcommand*{\lecomplete}{\le_\mathsf{cpl}}
\newcommand*{\gecomplete}{\ge_\mathsf{cpl}}

\newcommand*{\by}[1]{(\text{#1})}

\newcommand*{\dtv}{d_\mathsf{TV}}

\renewcommand{\epsilon}{\varepsilon}
\renewcommand{\phi}{\varphi}

\newcommand*{\revv}[1]{#1^\bullet}
\newcommand*{\grv}[1]{#1_\bullet}

\DeclareSymbolFont{symbolsA}{U}{txsya}{m}{n}
\DeclareSymbolFont{symbolsC}{U}{txsyc}{m}{n}
\DeclareMathSymbol{\multimapdot}{\mathrel}{symbolsC}{20}
\DeclareMathSymbol{\multimapdotinv}{\mathrel}{symbolsC}{21}

\makeatletter
\newcommand{\dashedrightarrow}{\mathrel{\mathpalette\dashedrightarrow@\relax}}
\newcommand{\dashedrightarrow@}[2]{%
  \begingroup
  \settowidth{\dimen\z@}{$\m@th#1\rightarrow$}%
  \settoheight{\dimen\tw@}{$\m@th#1\rightarrow$}%
  \sbox\z@{%
    \makebox[\dimen\z@][s]{%
      \clipbox{0 0 {0.4\width} 0}{\resizebox{\dimen\z@}{\height}{$\m@th#1\dashrightarrow$}}%
      \hss
      \clipbox{{0.69\width} {-0.1\height} 0 {-\height}}{$\m@th#1\rightarrow$}%
    }%
  }%
  \ht\z@=\dimen\tw@ \dp\z@=\z@
  \box\z@
  \endgroup
}
\makeatother

\newtheorem{thm}{Theorem}[section]
\newtheorem{lem}[thm]{Lemma}
\newtheorem{cor}[thm]{Corollary}
\newtheorem{prop}[thm]{Proposition}

\theoremstyle{definition}
\newtheorem{defn}[thm]{Definition}
\newtheorem{expl}[thm]{Example}
\newtheorem{rem}[thm]{Remark}
\newtheorem{conv}[thm]{Convention}

\begin{document}

\title{Behavioural Conformances based on Lax Couplings
}

\author{
  \IEEEauthorblockN{Paul Wild}
  \IEEEauthorblockA{Friedrich-Alexander-Universität Erlangen-Nürnberg\\
  Germany \\
  paul.wild@fau.de}
\and
  \IEEEauthorblockN{Lutz Schröder}
  \IEEEauthorblockA{Friedrich-Alexander-Universität Erlangen-Nürnberg\\
  Germany \\
  lutz.schroeder@fau.de}
}

\maketitle

\begin{abstract}
  Behavioural \emph{conformances} -- e.g.\ behavioural equivalences,
  distances, preorders -- on a wide range of system types
  (non-deterministic, probabilistic, weighted etc.)  can be dealt with
  uniformly in the paradigm of universal coalgebra. One of the most
  commonly used constructions for defining behavioural distances on
  coalgebras arises as a generalization of the well-known Wasserstein
  metric.  In this construction, couplings of probability
  distributions are replaced with couplings of more general objects,
  depending on the functor describing the system type.  In many cases,
  however, the set of couplings of two functor elements is empty,
  which causes such elements to have infinite distance even in
  situations where this is not desirable.  We propose an approach to
  defining behavioural distances and preorders based on a more liberal
  notion of coupling where the coupled elements are matched laxly
  rather than on-the-nose. We thereby substantially broaden the range
  of behavioural conformances expressible in terms of couplings,
  covering, e.g., refinement of modal transition systems and
  behavioural distance on metric labelled Markov chains.
\end{abstract}


\section{Introduction}
\label{sec:intro}

\noindent In the analysis of state-based systems, different methods of
modelling and comparing the behaviour of states may be appropriate
depending on the system type.  For systems featuring discrete or
qualitative data, such as basic labelled transition systems, one
typically considers some form of \emph{bisimilarity} which renders
states as either equivalent in behaviour or not.  For systems with
continuous or quantitative data, such as Markov chains or metric
transition systems, one instead often works with \emph{behavioural
  distances}, which give a quantitative account of how equivalent or
in equivalent the behaviour of states is
(e.g.~\cite{GiacaloneEA90,BreugelWorrell05,AlfaroEA09}).  The latter
have been the subject of a considerable amount of research in recent
years, a significant part of which deals with the problem of defining
and characterizing behavioural distances at the level of universal
coalgebra~\cite{Rutten00}, a general framework in which state-based
systems of various types can be uniformly treated by modelling system
types as functors on the category of sets.

A common problem in the coalgebraic treatment of behavioural distances is that of lifting the functor encapsulating the system type from the category of sets to the category of (pseudo-)metric spaces.
Both in concrete examples and in the coalgebraic treatment~\cite{BBKK18}, there are essentially two known constructions of pseudometric liftings, which respectively generalize the two sides of the well-known Kantorovich-Rubinstein duality for distances of probability distributions~cite{Kantorovich39}.
The first of these generalizes the description of the distance of two probability distributions on a pseudometric space via differences between expected values of non-expansive real-valued predicates, and has been termed the \emph{Kantorovich lifting}~\cite{BBKK18} or the \emph{codensity lifting}~\cite{KomoridaEA21}. This lifting instantiates  also to a number of other constructions on pseudometric spaces, such as the \emph{Hausdorff distance} on the powerset, and in fact it has been shown that under mild conditions, \emph{every} pseudometric lifting can be seen as an instance of the codensity lifting~\cite{KantorovichFunctors}.

The story is much less clear for the second construction, called the \emph{Wasserstein lifting} or simply the \emph{coupling-based lifting}.
In this construction, the distance of two elements of a set $FX$ (where~$F$ is a set functor and~$X$ is the underlying set of a pseudometric space) is defined in terms of the set of their \emph{couplings}, which are elements of $F(X\times X)$ that map to the original elements when $F$ is applied to the left and right projections.
This generalizes the idea of probabilistic couplings which have two given distributions as their \emph{marginal distributions}.
One major issue with this generalization is that for many functors, such couplings may in general fail to exist, which causes the distances between the corresponding elements to be infinite.
For instance, already when generalizing from probability measures to measures with arbitrary total mass, couplings only exist for measures with equal total mass.
Similarly, two lists  admit a coupling only if they have the same length.
In both of these cases, simply declaring the distance between differently-sized objects to be infinite does not tell the full story, however; after all, it makes a difference whether two lists only differ by one element or by thousands of them, and such differences should be reflected by the behavioural distance.

In the present work, we propose a variant of the coupling-based lifting that maintains the same categorical level of generality while addressing the issue of potentially non-existent couplings by working with couplings whose projections need not match the given elements exactly, but instead take into account any mismatches that occur in this way in the computation of the distance value.
Technically, this is achieved by working with functors~$F$ from the category of sets to the category of (pseudo-)metric spaces, where the pseudometric on $FX$ prescribes how much a mismatch in the coupling should contribute to the overall distance value.
In fact, similar reasoning applies when one considers preorders instead of pseudometrics, where now the projections of the coupling need only satisfy inequalities with the given elements rather than strict equalities.
This kind of setup makes it possible to capture notions of \emph{similarity}.
In order to accommodate both of these settings in a uniform manner we parametrize our setup over the choice of a \emph{quantale}~$\V$, and then work at the level of $\V$-categories, thus covering general \emph{behavioural conformances}. For instance, for~$\V=[0,\infty]$, $\V$-categories are (extended) pseudo-metrics, or more generally asymmetric distances, while for $\V=2$, $\V$-categories are equivalence relations or preorders, respectively (depending on whether or not symmetry is imposed on $\V$-categories).

We center the technical development around \emph{$\V$-valued lax
  extensions}~\cite{Hofmann07}, which serve as a tool to construct
functor liftings and behavioural conformances systematically.
Additionally, lax extensions provide native support for conformances
and in fact notions of (bi-)simulation between two potentially
different coalgebras.  As lax extensions usually apply to set
(endo)functors, we need to adapt the usual notion slightly to
encompass non-endofunctors of the type discussed earlier.  Following
the lines of previous work on coalgebraic
similarity~\cite{BilkovaEA13,KurzVelebil16}, we are thus led to study
\emph{$\V$-modules}~\cite{Street80}, also known as $\V$-valued \emph{
  distributors}~\cite{Benabou00}.

We then proceed to apply the general construction to a number of
different examples.  These include the aforementioned cases of
measures of varying total mass and strings of varying length, but also
behavioural distances for metric transition systems and
metric-labelled Markov chains, as well as more general weighted
transition systems.  On the two-valued side, we cover various notions
of simulation, including ready and complete simulation of labelled
transition systems and \emph{refinement} of \emph{label-structured
  modal transition systems}~\cite{BauerEA12}.

\subsection{Related Work}

\noindent We have already mentioned general work on coalgebraic behavioural distances~\cite{BBKK14,BBKK18}; in concrete shape, behavioural distances go back to work on probabilistic systems~\cite{LarsenSkou91,GiacaloneEA90,DesharnaisEA04,BreugelWorrell05}.
While these original results focus on liftings of functors to metric or pseudometric spaces, more recently there has been a focus on generalizing beyond pseudometrics  by working with quantales and quantitative relations~\cite{Gavazzo18,FuzzyLaxHemi,QuantifiedVB}, fibrations~\cite{KomoridaEA21,BonchiEA23}, and functors native to categories of spaces other than sets~\cite{QuantalicHM,KantorovichFunctors}.

The coalgebraic treatment of bisimulations and bisimilarity through liftings to the category of relations goes back to Barr~\cite{Barr70}, and has inspired a large body of work on relation liftings for both bisimulations and simulations:
The idea of modelling simulations using functors from sets to preorders was already proposed by Hughes and Jacobs~\cite{HughesJacobs04}, and was subsequently expanded by others~\cite{Hermida11,BilkovaEA13,KurzVelebil16}; some of the later work already features categories of $\V$-distributors.
In addition to Barr's relation lifting, which is based on representing relations as spans, other types of relation liftings have been proposed under different names (and with some variations in the definitions), most commonly \emph{relators}~\cite{Thijs96,Levy11} and \emph{lax extensions}~\cite{HughesJacobs04,MartiVenema15}.

The quantalic Wasserstein lax extension in fact predates its pseudometric counterpart, already appearing in work on topological theories~\cite{Hofmann07}.
More recently, it has been applied to up-to techniques~\cite{BonchiEA23}.
The Kantorovich lax extension is comparatively younger, first appearing in work that shows that under certain conditions every lax extension can be cast as an instance of the Kantorovich construction~\cite{FuzzyLax,FuzzyLaxHemi}.

Both of the original formulations of the metric for probability measures are due to Kantorovich~\cite{Kantorovich39} and his work forms the root of the field of \emph{transportation theory}; Villani~\cite{Villani09} provides an overview.
The problem of assigning distances to measures of arbitrary mass has also come up quite early~\cite{Dantzig51}, under the name of \emph{inhomogeneous} or \emph{unbalanced optimal transport}.
Various solutions to this problem have been suggested since then, many of them quite recent thanks to its applicability in the domain of artificial intelligence~\cite{ChizatEA18b,LieroEA18}.
Another recent application of unbalanced optimal transport appears in work on bisimulations for dynamical systems~\cite{BacciBacciEA24}.

\subsection{Organization}

\noindent We recall some concepts from universal coalgebra and the theory of quantales in Section~\ref{sec:prelim}.
In particular, this includes the relevant categories of quantale-valued relations, categories and distributors.

In Section~\ref{sec:functors}, we discuss some motivating examples of functors from sets into preorders or (hemi-)metric spaces and the notions of conformance that can be derived from them.
We continue discussing these examples throughout the remainder of the paper.

In Section~\ref{sec:lax}, we introduce a notion of \emph{distributorial} lax extension that applies to functors of this type and show under which conditions the Kantorovich (codensity) construction is an instance of this notion.

In Section~\ref{sec:exact}, we recall the notion of an exact square and identify preservation of such squares as the correct condition to place on a functor in order for it to admit a Wasserstein extension.

In Section~\ref{sec:wasserstein}, the distributorial Wasserstein extension is introduced; we fully work out the conditions under which it is a distributorial lax extension, give some examples and compare it with its non-distributorial counterpart.

In Section~\ref{sec:examples}, we conclude the discussion of the running examples from Section~\ref{sec:functors} and complement them with additional examples, including the aforementioned notions of simulation and a distance for measures of varying mass.

In Section~\ref{sec:duality}, we discuss some examples of distributorial lax extensions which admit representations as both Kantorovich and Wasserstein extensions, in analogy with the Kantorovich/Rubinstein duality for metrics between probability distributions.

Section~\ref{sec:conclusion}, finally, contains some concluding remarks and discusses some directions for future research.

Proofs are mostly omitted or only sketched.
Full proofs can be found in the appendix.

\section{Preliminaries}
\label{sec:prelim}

\subsection{Coalgebra}

\noindent
Universal coalgebra~\cite{Rutten00} enables the unified study of a large range of types of state-based systems, including labelled transition systems, many classes of automata, and probabilistic systems such as Markov chains.
This uniformity is achieved by casting such systems as \emph{(functor) coalgebras}, which are maps of type $\alpha\colon X\to FX$ where $X$ is the set of states of the system and the functor $F\colon\Set\to\Set$ describes the branching type of the system.

As an example, a \emph{labelled transition system (LTS)} with transition relation ${\rightarrow}\subseteq X\times\A\times X$ can be seen as coalgebra
by representing $\rightarrow$ as a map $\alpha\colon X\to\Pow(\A\times X)$, where $\Pow$ denotes powerset.
The associated functor $\Pow(\A\times(-))$ maps each set $X$ to the set of subsets of $\A\times X$, so that $\alpha$ associates each state with the set of pairs of label and successor possible from this state.
On functions, the functor $\Pow(\A\times(-))$ acts by image: $\Pow(\A\times f)(U) = \{(a,f(x)) \mid (a,x)\in U\}$.

Similarly, a \emph{Markov chain} may be rendered as a coalgebra $\alpha\colon X\to\dfun X$ for the finite distribution functor $\dfun$, which maps each set $X$ to the set $\dfun X$ of finitely supported probability distributions on $X$, i.e.~to functions $\mu\colon X\to[0,1]$ such that $\supp(\mu) \coloneqq \{x\in X \mid \mu(x) > 0\}$ is finite and $\sum_{x\in X} \mu(x) = 1$.
For convenience, we write $\mu(U) = \sum_{x\in U} \mu(x)$.
On functions, the functor $\dfun$ acts by direct image: $\dfun f(\mu)(y) = \mu(f^{-1}[y])$, where $f^{-1}[y]$ is the preimage of $y$ under $f$.

\subsection{Quantales}

\noindent
In order to allow for a uniform treatment of both (bi)si\-mu\-la\-ti\-ons and behavioural distances, we work with quantales, which combine arithmetic and ordered structure.
Explicitly, a \emph{(commutative unital) quantale} is a tuple $(\V,\vle,\vtimes,\vunit)$, where $(\V,\vle)$ is a complete lattice and $(\V,\vtimes,\vunit)$ is a commutative monoid such that multiplication is \emph{join-continuous}:
\begin{equation*}
  \textstyle a \vtimes \vJoin_{i\in I} b_i = \vJoin_{i\in I} a \vtimes b_i.
\end{equation*}
Here, joins (or least upper bounds) are denoted by $\vJoin$, and similarly meets (or greatest lower bounds) are denoted by $\vMeet$.
The least and greatest element of $\V$ are denoted by $\vminel$ and $\vmaxel$, respectively.
Join-continuity of $a\vtimes(-)$ entails that it has a right adjoint $\hom(a,-)$, that is, $a\vtimes b \vle c \iff b\vle\hom(a,c)$, called the \emph{internal hom} of the quantale.

We consider two main examples of quantale in this paper.
The first is the \emph{Boolean quantale} (or \emph{two-valued quantale}) $\twoQ$, where $\V = \{0,1\}$, with $0 \vle 1$ and the monoid operation given by meet $\land$.
One easily verifies that $\hom(a,b) = a \rightarrow b$ is Boolean implication.

The second main example of interest is the quantale $\zeroinfQ$, where $\V = [0,\infty] = \mathbb{R}^{+} \cup \{\infty\}$ is the set of \emph{extended non-negative real numbers} (``extended'' because it includes~$\infty$).
The order of this quantale is the \emph{inverted} order of the real numbers, that is ${\vle}={\ge}$ and $0=\vmaxel$, $\infty=\vminel$.
The monoid operation is extended addition of real numbers, with $\infty + x = x + \infty = \infty$.
The internal hom is given by \emph{truncated subtraction}: $\hom(a,b) = b\ominus a = \max(b-a, 0)$.

\begin{rem}\label{rem:other-quantale}
  Another related quantale that is often considered is based on the unit interval $[0,1]$, with truncated addition $a \oplus b = \min(a+b,1)$ as the monoid operation.
\end{rem}

\noindent
For a fixed quantale $\V$, a \emph{$\V$-valued relation} (or simply \emph{$\V$-relation}) between two sets $X$ and $Y$ is a map $r\colon X\times Y\to\V$.
In this case we also write $r\colon\frel{X}{Y}$.
Clearly, $\twoQ$-relations are just ordinary relations, while we sometimes refer to $\zeroinfQ$-relations as \emph{fuzzy relations}.

There are several important constructions on $\V$-relations.
Given $r\colon\frel{X}{Y}$ and $s\colon\frel{Y}{Z}$ and $f\colon X\to Y$,
\begin{itemize}
  \item the \emph{(relational) composite} $s\cdot r\colon\frel{X}{Z}$ is given by
    \begin{equation*}
      (s\cdot r)(x,z) = \textstyle\vJoin_{y\in Y} r(x,y) \vtimes s(y,z);
    \end{equation*}
  \item the \emph{converse} $\rev{r}\colon\frel{Y}{X}$ is given by $\rev{r}(y,x) = r(x,y)$;
  \item the \emph{graph} $\gr{f}\colon\frel{X}{Y}$ is given by $\gr{f}(x,y) = \vunit$ if $f(x) = y$ and $\gr{f}(x,y) = \vminel$ otherwise.
\end{itemize}

\begin{conv}
  For convenience, we often do not notationally distinguish between a function $f\colon X\to Y$ and its graph $\gr{f}\colon\frel{X}{Y}$.
\end{conv}

The graph of an identity function $\id_X\colon X\to X$ is also denoted by $\Delta_X\colon\frel{X}{X}$, and it is easily checked that the~$\Delta_X$ are the neutral elements for composition of $\V$-relations.
In fact, $\V$-relations form a category $\VRel$, which has sets as objects and $\V$-relations between them as arrows.
This category is order-enriched, because $\V$-relations between two given sets can be compared pointwise:
\begin{equation*}
  r \vle r' \iff \forall x, y.\; r(x,y) \vle r'(x,y)
\end{equation*}

\noindent
A $\V$-relation $d_X\colon \frel{X}{X}$ is called a \emph{$\V$-category} if $\Delta_X\vle d_X$ and $d_X\cdot d_X\vle d_X$.
A $\V$-category is \emph{symmetric} if $\rev{d_X} \vle d_X$, and \emph{separated} if $\vunit \vle d_X(x,x')$ implies $x = x'$ for all $x,x'\in X$.

In the Boolean case, we note that $\twoQ$-categories correspond to \emph{preorders} and symmetric $\twoQ$-categories to \emph{equivalence relations}.
In the real-valued case, the two inequalities of $\zeroinfQ$-categories correspond to reflexivity and the triangle inequality respectively, so that $\zeroinfQ$-categories are exactly \emph{hemimetrics}, also known as \emph{generalized metric spaces}~\cite{Lawvere73}.
Similarly, a symmetric $\zeroinfQ$-category is a \emph{pseudometric} and a symmetric and separated $\zeroinfQ$-category are \emph{metric}.

The corresponding notion of structure-preserving map between $\V$-categories is that of a $\V$-functor.
Given $\V$-categories $(X,d_X)$ and $(Y,d_Y)$, a function $f\colon X\to Y$ is  a \emph{$\V$-functor} if for all $x,x'\in X$ we have $d_X(x,x') \vle d_Y(f(x),f(x'))$.
This condition can equivalently expressed in a pointfree manner as $f\cdot d_X \vle d_Y\cdot f$.
Together, $\V$-categories and $\V$-functors form a category $\VCat$.
In the special cases of $\twoQ$ and $\zeroinfQ$, we also sometimes denote this category by $\PreOrd$ or $\HMet$, respectively.

One important example of a $\V$-category is given by $\V$ itself:
One easily verifies that $\hom\colon\frel{\V}{\V}$ is a $\V$-category.

The final category of interest arises by considering $\V$-relations between $\V$-categories that respect the structures of the given $\V$-categories.
This can be expressed in multiple equivalent ways:

\begin{lem}\label{lem:v-dist}
  Let $(X,d_X)$ and $(Y,d_Y)$ be $\V$-categories and $r\colon\frel{X}{Y}$.
  Then the following are equivalent:
  \begin{enumerate}
    \item We have $r\cdot d_X \vle r$ and $d_Y\cdot r \vle r$.
    \item We have $r\cdot d_X = r$ and $d_Y\cdot r = r$.
    \item The map $r\colon (X,\rev{d_X})\times(Y,d_Y) \to (\V,\hom)$ is a $\V$-functor.
    \item $(X+Y,s)$ is a $\V$-category, where $s(x,x') = d_X(x,x')$, $s(x,y) = r(x,y)$, $s(y,x) = \vminel$, $s(y,y') = d_Y(y,y')$ for all $x,x'\in X$ and $y,y'\in Y$.
  \end{enumerate}
\end{lem}

\noindent
We say that $r$ is a \emph{$\V$-distributor} if it satisfies the equivalent conditions of Lemma~\ref{lem:v-dist}.
Other names for them include \emph{$\V$-profunctor}, \emph{$\V$-bimodule} or \emph{$\V$-module}.
$\V$-distributors form a category $\VDist$, in which the objects are $\V$-categories, the arrows are $\V$-distributors between them, and identity arrows are the $\V$-category structures.

\section{Functors with structure}\label{sec:functors}

\noindent As indicated in the introduction, our construction takes
functors $F\colon\Set\to\VCat$ as input.  We sometimes write
$|F|\colon\Set\to\Set$ for the functor that arises by forgetting the
$\V$-categorical structure, i.e.\ the composite of~$F$ with the
forgetful functor $\VCat\to\Set$.  Similarly, we may use
$|FX|$ to denote the underlying set of the $\V$-category
$FX$.  Depending on the quantale, we use $\le_{FX}$ or
$d_{FX}$ to denote the associated
$\V$-relation, that is $\le_{FX},d_{FX}\colon |FX|\times|FX|\to\V$.

The motivation behind using functors of this type is that, as usual, $FX$ describes the space of possible successor structures of states in a coalgebra, but that this space already has some $\V$-categorical structure defined on it.
If we then lift such a functor $F$ to an endofunctor $\overline{F}\colon\VCat\to\VCat$, the structure on $\overline{F} (X,d_X)$ combines the structure $d_X$ on $X$ with that on $FX$.
Similarly, we can define a notion of simulation between coalgebras by combining the data from the relation between states with the data about their successor structures.
We discuss some examples to illustrate this idea:

\begin{expl}\label{expl:simulations}
  Fix a set $\A$ of labels and let $|F| = \Pow(\A\times(-))$ be the
  LTS functor (cf.\ Section~\ref{sec:prelim}).  We put
  $\le_{FX} \coloneqq {\subseteq}$.  Coalgebraically, the usual notion
  of bisimulation is given via the \emph{Egli-Milner construction}:
  For $r\colon\frel{X}{Y}$ define $Lr\colon\frel{FX}{FY}$ to be the
  relation where $U\,Lr\,V$ iff for all $(a,x)\in U$ there exists
  $(a,y)\in V$ such that $x\,r\,y$ and for all $(a,y)\in V$ there
  exists $(a,x)\in U$ such that $x\,r\,y$.  Then, given two LTS
  $\alpha\colon X\to FX$ and $\beta\colon Y\to FY$, a relation
  $r\colon\frel{X}{Y}$ is a \emph{bisimulation} iff
  $r \vle \rev{\beta}\cdot Lr\cdot\alpha$, equivalently
  $\beta\cdot r\vle Lr\cdot\alpha$.

  To obtain a notion of \emph{simulation}, we simply replace $Lr$ with
  \begin{equation*}
    L^\subseteq r \coloneqq {\le}_{FY}\cdot Lr\cdot {\le}_{FX},
  \end{equation*}
  which has the same effect as simply omitting the second clause from
  the definition of $Lr$ above (see the appendix).
\end{expl}

\begin{expl}\label{expl:lists}
  Next, we consider the \emph{list functor} $\List$, which maps each set to the set of finite sequences over it:
  \begin{math}
    \List X = \{(x_1,\dots,x_n) \mid n < \omega, x_i\in X \text{ for each } i\}.
  \end{math}
  Given a pseudometric $d_X$ on a set $X$, where the distance $d_X(x,x')$ is thought of as the cost of transforming the symbol $x$ into the symbol $x'$,
  we obtain a generalized \emph{Hamming distance} $\Ham(d_X)$ on $\List X$, where lists of different length are assigned distance $\infty$ and the distance of lists of equal length is the element-wise sum of distances.
  
  We can obtain non-infinite distance values for lists of different lengths by extending $\List$ to a functor $F\colon\Set\to\HMet$ such that $|F| = \List$.
  Explicitly, we put $d_{FX}(s,t) = \left||t|-|s|\right|$ if one of $s$ and $t$ is a subsequence of the other, and $d_{FX}(s,t) = \infty$ otherwise.
  We now define $\Lev(d_X) \coloneqq d_{FX}\cdot\Ham(d_X)\cdot d_{FX}$ and obtain a generalized form of \emph{Levenshtein distance} (or \emph{edit distance}), where the distance between two lists is defined to be the minimal total cost of changing the first list to the second through character insertions/deletions (modelled by $d_{FX}$) and character substitutions (modelled by $\Ham(d_X)$).
\end{expl}

\begin{expl}\label{expl:mlmc}
  Fix a metric space $(\A,d_\A)$ of labels and consider the functor $G = \dfun(\A\times(-))$.
  Coalgebras for~$G$ are \emph{metric-labelled Markov chains}.
  If $(X,d_X)$ is a metric space, then the distance of two distributions in $GX$ can be computed using the Wasserstein lifting applied to the product metric $d_\A\times d_X$.
  This construction cannot be achieved using exact couplings, however, as elements of $GX$ only admit couplings if they have the same probability mass for each individual label in $\A$.
  We will later see how one can use lax couplings to represent this construction, using the functor $F\colon\Set\to\HMet$ with $|F| = G$ and $d_{FX} = \Kant(d_\A\times\Delta_X)$ that partially applies the Wasserstein lifting to the metric on the labels while treating the set $X$ as a discrete metric space.
\end{expl}

We will use these functors as running examples throughout the paper.
Further examples are discussed in Section~\ref{sec:examples}.

\section{Distributorial lax extensions}
\label{sec:lax}

\noindent The further technical development will be centered around
variants of the notion of lax extension, which we use to define both
functor liftings and notions of (bi-)simulation.  First, we recall the
standard notion~\cite{Thijs96,SchubertSeal08,Levy11,MartiVenema15}:

\begin{defn}[Lax extension]\label{defn:lax}
  Let $F\colon\Set\to\Set$ be a functor.
  A \emph{($\V$-valued) lax extension} of $F$ is an assignment~$L$ from $\V$-valued relations $r\colon\frel{A}{B}$ to $\V$-valued relations $Lr\colon\frel{FA}{FB}$ satisfying the following properties for all $r,r'\colon\frel{A}{B}$, $s\colon\frel{B}{C}$ and all $f\colon A\to B$:
  \begin{description}
    \item[(L1)] $r \vle r' \implies Lr \vle Lr'$
    \item[(L2)] $Ls \cdot Lr \vle L(s \cdot r)$
    \item[(L3)] $Ff \vle Lf$ and $\rev{(Ff)} \vle L(\rev{f})$
  \end{description}
\end{defn}
\noindent 
A classical example of a lax extension in the case $\V=\zeroinfQ$ is
the \emph{Kantorovich/Wasserstein extension}
$\Kant$~\cite{Kantorovich39,BBKK18} of the discrete distribution
functor $\dfun$, whose coupling-based description is given as follows.
For a fuzzy relation $r\colon\frel{X}{Y}$ and distributions
$\mu\in\dfun X$ and $\nu\in\dfun Y$, we have
$\Kant(r)(\mu,\nu) = \inf_\rho \expect{\rho} r$, where $\expect{}$
takes expected values and the infimum ranges over all probability
distributions~$\rho$ on $X\times Y$ that have~$\mu$ and~$\nu$ as
\emph{marginals}, that is, $\mu(x) = \sum_{y\in Y} \rho(x,y)$ for each
$x\in X$ and $\nu(y) = \sum_{x\in X} \rho(x,y)$ for each $y\in Y$.
Such a~$\rho$ is called a \emph{coupling} of $\mu$ and $\nu$.

A lax extension induces a notion of $\V$-valued behavioural distance
on coalgebras:

\begin{defn}\label{defn:L-simulation}
  Let $L$ be a lax extension of $F\colon\Set\to\Set$, and let $\alpha\colon X\to FX$ and $\beta\colon Y\to FY$ be $F$-coalgebras.
  \begin{enumerate}
    \item\label{item:L-sim} A relation $r\colon\frel{X}{Y}$ is an \emph{$L$-simulation} if  $r \vle \rev{\beta}\cdot Lr\cdot\alpha$.
    \item\label{item:L-bd} \emph{$L$-behavioural distance} $\bd^L_{\alpha,\beta}\colon\frel{X}{Y}$ is the greatest $L$-simulation, that is
      \begin{equation*}
        \bd^L_{\alpha,\beta} \coloneqq \vJoin \{ r\colon\frel{X}{Y} \mid r \vle \rev{\beta}\cdot Lr\cdot\alpha\}.
      \end{equation*}
  \end{enumerate}
\end{defn}
\noindent The second clause above relies on the Knaster-Tarski
fixpoint theorem, together with the fact that the assignment
$r \mapsto \rev{\beta}\cdot Lr\cdot\alpha$ is monotone by the axioms
of lax extensions.

Lax extensions induce functor liftings.  Specifically, if
$F\colon\Set\to\Set$ is a set functor, then given a $\V$-category
$d_X\colon\frel{X}{X}$, the axioms of lax extensions guarantee that
$Ld_X\colon\frel{FX}{FX}$ is also a $\V$-category.  Therefore, we have
a functor $\overline{F}\colon\VCat\to\VCat$ given by
\begin{equation*}
  \overline{F}(X,d_X) = (FX,Ld_X)
  \quad\text{and}\quad
  \overline{F}f = Ff.
\end{equation*}
(It is easily verified that $\overline{F}$ preserves $\V$-functors.)

A suitable generalization of lax extensions to functors as considered in Section~\ref{sec:functors} arises from the following observation:
\begin{lem}\label{lem:distrib-lax}
  Let $F\colon\Set\to\VCat$ be a functor, and let $L\colon\VRel\to\VRel$ be a lax extension of $|F|\colon\Set\to\Set$.
  Then the following are equivalent:
  \begin{description}
    \item[(D)]\label{item:distrib-lax-diag}
      For every set $X$, $d_{FX} \vle L\Delta_X$.
    \item[(D')]\label{item:distrib-lax-distrib} For every
      $r\colon\frel{X}{Y}$, $Lr\colon\frel{FX}{FY}$ is a $\V$-distributor (cf.\
      Section~\ref{sec:prelim}).
  \end{description}
\end{lem}

\begin{defn}[Distributorial lax extension]\label{defn:distrib-lax}
  Let $F\colon\Set\to\VCat$ be a functor, and let $L\colon\VRel\to\VRel$ be a lax extension of $|F|\colon\Set\to\Set$.
  We say that $L$ is a \emph{distributorial lax extension} of $F$ if it satisfies the equivalent conditions of Lemma~\ref{lem:distrib-lax}.
\end{defn}

\begin{rem}\label{rem:lax-functor}
  By Lemma~\ref{lem:distrib-lax}, just like a lax extension of a functor $F\colon\Set\to\Set$ is a lax functor $L\colon\VRel\to\VRel$ (in the sense that (L2) says the composition is preserved laxly), a distributorial lax extension of a functor $F\colon\Set\to\VCat$ is a lax functor $L\colon\VRel\to\VDist$.
\end{rem}
Classically, lax extensions often arise as instances of two generic
constructions, the so-called Kantorovich (or codensity) and
Wasserstein (or coupling-based) extensions~\cite{BBKK18,FuzzyLaxHemi},
which are both based on a choice of $\V$-valued predicate
liftings~\cite{CirsteaEA11}.
\begin{defn}
  Let $F\colon\Set\to\Set$ be a functor.
  A (unary) \emph{$\V$-valued predicate lifting} is a natural transformation
  \begin{equation*}
    \lambda\colon\Set(-,\V)\to\Set(F-,\V).
  \end{equation*}
  We say that $\lambda$ is \emph{monotone} if for every $f,g\colon X\to\V$ we have that $f\vle g$ implies $\lambda_X(f) \vle \lambda_X(g)$.
\end{defn}
\noindent As indicated by the name, the components $\lambda_X$ of a
predicate lifting $\lambda$ lift $\V$-valued predicates on $X$ to
$\V$-valued predicates on $FX$.

Under suitable conditions on the functor and the predicate lifting, both the Kantorovich and the Wasserstein construction generalize to distributorial lax extensions.
For the former this is comparatively straightforward.
Recall that for every $\V$-valued relation $r\colon\frel{X}{Y}$ and every $\V$-valued predicate $f\colon X\to\V$, the \emph{relational image} $r[f]\colon Y\to\V$ of $f$ under $r$ is given by
\begin{equation*}
  r[f](b) = \vJoin_{a\in A} f(a)\vtimes r(a,b).
\end{equation*}
\begin{defn}[Kantorovich extension]\label{defn:kantorovich}
  Let $F\colon\Set\to\VCat$ be a functor, and $\lambda$ a monotone predicate lifting for $|F|$.
  The \emph{Kantorovich extension} $K_\lambda$ is given by
  \begin{equation*}
    K_\lambda r(t_1,t_2) = \vMeet_{\mathclap{f\colon X\to\V}} \hom(\lambda_X(f)(t_1), \lambda_Y(r[f])(t_2)).
  \end{equation*}
  Given a set $\Lambda$ of predicate liftings, put $K_\Lambda r \coloneqq \vMeet_{\lambda\in\Lambda} K_\lambda r$.
\end{defn}
\begin{thm}\label{thm:k-lax}
  Let $F\colon\Set\to\VCat$ be a functor and $\lambda$ a monotone predicate lifting for $|F|$.
  Then the Kantorovich extension is a distributorial lax extension iff $\lambda$ is a natural transformation
  \begin{equation*}
    \lambda\colon\Set(-,\V)\to\VCat(F-,(\V,\hom)).
  \end{equation*}
\end{thm}
\noindent (In words, the condition of the theorem says that $\lambda$
lifts predicates on~$X$ to $\V$-functorial predicates on~$FX$.)
\begin{proof}
  As the definition is the same as usual, $K_\lambda$ satisfies the
  axioms (L1)-(L3) by established results~\cite[Theorem
  5.8]{FuzzyLaxHemi}.  We thus only need to check that condition~(D)
  of Lemma~\ref{lem:v-dist} is equivalent to the condition on~$\lambda$.  Expanding definitions, and noting that
  $\Delta_X[f] = f$, (D) is equivalent to
  \begin{equation*}
    d_{FX}(t_1,t_2) \vle \hom(\lambda_X(f)(t_1), \lambda_X(f)(t_2))
  \end{equation*}
  for every $X$, $f$, $t_1$ and $t_2$, which is precisely $\V$-functoriality of $\lambda_X(f)$.
\end{proof}

\begin{rem}
  The Kantorovich extension admits a point-free representation, based on the observation that a $\V$-valued predicate $f\colon X\to\V$ can be seen as a $\V$-relation $f\colon\frel{1}{X}$.
  The relational image of $f$ under $r$ is then simply given by the relational composition $r\cdot f$, and the Kantorovich extension can be given in terms of the right adjoint $-\multimapdotinv s$ of relational composition $-\cdot s$:
  \begin{equation*}
    K_\lambda r = \vMeet_{\mathclap{f\colon\frel{1}{X}}} \lambda_X(f) \multimapdotinv \lambda_Y(r\cdot f).
  \end{equation*}
  This representation is based on recent similar one where $\V$-relations of the form $\frel{X}{1}$ are used instead~\cite{PointFreeLax}.
\end{rem}
In the following two sections, we introduce a distributorial version of the other generic construction, the Wasserstein extension, and work out conditions under which it is a distributorial lax extension.

\section{Exact squares}\label{sec:exact}

\noindent The Wasserstein extension~\cite{HofmannEA14,BBKK18,FuzzyLaxHemi} of a functor is  a lax extension only if the functor at hand preserves weak pullbacks and the corresponding predicate lifting satisfies certain requirements.
We will now adapt these requirements to the distributorial setting, starting with those on the functor.

A weak pullback is a weak limit of a cospan diagram, that is, a diagram
\begin{equation}\label{eqn:diag-pxyz}
  \begin{tikzcd}
    P \arrow[d, "u"] \arrow[r, "v"] & Y \arrow[d, "g"] \\
    X \arrow[r, "f"]                & Z               
  \end{tikzcd}    
\end{equation}
that commutes and additionally satisfies the universal property that for every triple $(Q, s\colon Q\to X, t\colon Q\to Y)$ satisfying $f\cdot s = g\cdot t$ there exists a (not necessarily unique) arrow $h\colon Q\to P$ such that $s = u\cdot h$ and $t = v\cdot h$.
It can be observed that the above diagram commutes iff $\rev{g}\cdot f \supseteq v\cdot\rev{u}$ and is a weak pullback iff $\rev{g}\cdot f = v\cdot\rev{u}$, so that weak pullback preservation by $F$ may be phrased as the condition $\rev{g}\cdot f = v\cdot\rev{u} \implies \rev{(Fg)}\cdot Ff \subseteq Fv\cdot\rev{(Fu)}$ (noting that the reverse inclusion always holds by functoriality).

A square satisfying the condition $\rev{g}\cdot f=v\cdot\rev{u}$ is also known as a \emph{Beck-Chevalley square} or an \emph{exact square}~\cite{Guitart80}.
We may more appropriately write this condition as $\rev{g}\cdot\gr{f}=\gr{v}\cdot\rev{u}$, using the two embeddings of functions into $\V$-relations, the covariant functor $\gr{(-)}\colon\Set\to\VRel$ and the contravariant functor $\rev{(-)}\colon\Set^\op\to\VRel$, which act as identity on objects and send maps to their graphs and inverse graphs, respectively.

The analogue of exact squares in this sense in the category $\VDist$~\cite{BilkovaEA13,KurzVelebil16} is based on the functors $\grv{(-)}\colon\VCat\to\VDist$ and $\revv{(-)}\colon\VCat^\op\to\VDist$, which also act as identity on objects and send a $\V$-functor $f\colon\frel{(X,d_X)}{(Y,d_Y)}$ to $\grv{f} \coloneqq d_Y\cdot f$ and $\revv{f} \coloneqq \rev{f}\cdot d_Y$, respectively.
To see that $\grv{f}$ is indeed a $\V$-distributor, note that $d_Y\cdot f\cdot d_X \vle d_Y\cdot d_Y\cdot f \vle d_Y\cdot f$, using that $f$ is a $\V$-functor and $d_Y$ is a $\V$-category, which shows both required inequalities at once.
The proof for $\revv{f}$ is very similar.
We now say that a commuting square~\eqref{eqn:diag-pxyz} in $\VCat$ is \emph{exact} if $\revv{g}\cdot\grv{f} \vle \grv{v}\cdot\revv{u}$.
Expanding definitions, a square of shape~\eqref{eqn:diag-pxyz} in $\VCat$ with $\V$-category structures $d_P$, $d_X$, $d_Y$ and $d_Z$ on $P$, $X$, $Y$ and $Z$ is thus exact iff
\begin{equation*}
  \rev{g} \cdot d_Z \cdot f = d_Y \cdot v \cdot \rev{u} \cdot d_X,
\end{equation*}
or, in pointful notation, iff for all $x\in X$ and $y\in Y$ we have
\begin{equation*}
  \textstyle d_Z(f(x),g(y)) = \vMeet_{p\in P} d_X(x,u(p)) \vtimes d_Y(v(p),y).
\end{equation*}

In analogy to the requirement that $\Set$ functors preserve weak
pullbacks we typically want our functors to preserve exact squares.
To this end, it is sometimes convenient to assume that the square to
be preserved is a pullback square instead of just a weak pullback
square (in analogy to the well-known fact that preservation of weak
pullbacks is equivalent to weak preservation of
pullbacks~\cite{GummSchroder00}):

\begin{lem}\label{lem:exact-via-pullback}
  Let $F\colon\Set\to\VCat$ be a functor. Then
  $F$ preserves exact squares iff for every pullback square~\eqref{eqn:diag-pxyz} we have $\revv{(Fg)}\cdot\grv{(Ff)} \vle \grv{(Fv)}\cdot\revv{(Fu)}$.
\end{lem}

\begin{expl}\label{expl:running-exact}
  The functors in our running examples preserve exact squares.
  We show this for the case of the list functor from Example~\ref{expl:lists} and defer the proofs for the other two functors to Section~\ref{sec:examples}.
  Assume a pullback square~\eqref{eqn:diag-pxyz}.
  Let $s\in FX$ and $t\in FY$.
  If $Ff(s)$ is a subsequence of $Fg(t)$, then for each entry $x$ in $s$ there is a corresponding entry $y$ in $t$ such that $f(x) = g(y)$ and these entries $y$ together form a subsequence of $t$.
  By assumption we can now form the sequence $r$ consisting of the pairs $(x,y)$ in $P$; then $Fu(r) = s$ and $Fv(r)$ is a subsequence of $t$.
  The case where $Fg(t)$ is a subsequence of $Ff(s)$ is symmetrical, and in the case where neither is a subsequence of the other nothing needs to be shown.
\end{expl}
\noindent
In all of the examples we consider in this paper, the functor at hand also satisfies the following condition:

\begin{defn}\label{defn:functional-bisim}
  A functor $F\colon\Set\to\VCat$ is said to be \emph{cool}
  if for every surjective function $f\colon X\to Y$ we have $d_{FY}\cdot Ff \vle Ff\cdot d_{FX}$ and $\rev{(Ff)}\cdot d_{FY} \vle d_{FX}\cdot\rev{(Ff)}$.
\end{defn}

\begin{rem}\label{rem:cool-equal}
  The two inequalities in Definition~\ref{defn:functional-bisim} in fact give rise to equalities, as their converses already follow from the fact that $F$ sends functions to $\V$-functors.  
\end{rem}

If all the structures $d_{FX}$ are symmetric, as is the case for functors into equivalence relations or pseudometrics, then we only need to show one of the conditions of Definition~\ref{defn:functional-bisim}:
\begin{lem}\label{lem:cool-sym}
  Let $F\colon\Set\to\VCat$ be a functor such that we have $\rev{d_{FX}} = d_{FX}$ for every set $X$.
  Then $F$ is cool iff for every surjective $f\colon X\to Y$ we have $d_{FY}\cdot Ff \vle Ff\cdot d_{FX}$.
\end{lem}

\begin{expl}\label{expl:running-cool}
  The functors in our running examples are cool:
  \begin{enumerate}[wide]
    \item For the functor $FX = (\pfun(\A\times X), \subseteq)$ from Example~\ref{expl:simulations}, let $U\subseteq\A\times X$ and $V'\subseteq\A\times Y$ such that $f[U]\subseteq V$.
      Then we have $U\subseteq U'$ for $U' \coloneqq f^{-1}[V']$.

      For the other inequality, let $V\subseteq\A\times Y$ and $U'\subseteq\A\times X$ such that $V\subseteq f[U']$, and put $U\coloneqq f^{-1}[V] \cap U'$.
      Then $U\subseteq U'$ by definition, and also 
      \begin{equation*}
        f[U] = f[f^{-1}[V] \cap U'] = f[f^{-1}[V]] \cap f[U'].
      \end{equation*}
    \item For the list functor $F$ from Example~\ref{expl:lists} we only need to show one inequality by Lemma~\ref{lem:cool-sym}.
      Let $s\in FX$, $t'\in FY$ and put $\epsilon = d_{FY}(Ff(s),t')$.
      If neither of $Ff(s)$ and $t'$ is a subsequence of the other, then $\epsilon = \infty$ and there is nothing to show.
      If $t'$ is a subsequence of $Ff(s)$, then define $s'$ to be the subsequence of $s$ that arises by picking the same indices as for the occurrence of $t'$ within $Ff(s)$.
      Finally, if $Ff(s)$ is a subsequence of $t'$, it suffices to define a supersequence $s'$ of $s$ such that $Ff(s') = t'$.
      This is easily achieved using surjectivity of $f$, making sure that $s$ matches $s'$ in the same indices as $Ff(s)$ matches $t'$.
    \item For the functor from Example~\ref{expl:mlmc}, see the appendix.
  \end{enumerate}
\end{expl}

\section{The Distributorial Wasserstein Extension}
\label{sec:wasserstein}

The distributorial Wasserstein extension is based on a single unary predicate lifting $\lambda$ of the functor $F\colon\Set\to\VCat$, of which we need to require additional properties.
As we work with pointfree representations wherever possible, the following notation will be helpful: given a $\V$-valued predicate $g\colon Y\to\V$, define $\diag{g}\colon\frel{Y}{Y}$ to be the endorelation on $Y$ where
$\diag{g}(y,y) = g(y)$ and $\diag{g}(y,y') = \vminel$ whenever $y\neq y'$.
We most often apply this construction to lifted predicates $\lambda_X(f)$ where $f\colon X\to\V$, for which we introduce the notation $\lambdadiag_X(f) \coloneqq \diag{\lambda_X(f)}$.
Additionally, for every set $Y$ we write $\vunit_Y\colon Y\to\V$ for the constant map $y\mapsto\vunit$ to the unit of the quantale.

\begin{defn}[Well-behaved predicate lifting]\label{defn:well-behaved}
  A predicate lifting $\lambda$ for a functor $F\colon\Set\to\VCat$ is said to be \emph{well-behaved} if it satisfies the following properties for all $f,g\colon X\to\V$:
  \begin{description}
    \item[Monotonicity:] $f \vle g \implies \lambda_X(f) \vle \lambda_X(g)$
    \item[Preservation of unit:] $\vunit_{FX} \vle \lambda_X(\vunit_X)$
    \item[$\V$-Normality:]
      \begin{math}
        \lambdadiag_X(g)\cdot d_{FX}\cdot\lambdadiag_X(f)
        \vle d_{FX}\cdot\lambdadiag_X(f\vtimes g)\cdot d_{FX}
      \end{math}
  \end{description}
\end{defn}
In the Boolean quantale $\twoQ$, $\V$-normality expands to the condition that for all $f,g\in 2^X$ and all $t_1,t_2\in FX$,
\begin{multline}\label{eqn:v-normal}
  t_1\vDash\lambda_X(f) \;\land\; t_1 \le_{FX} t_2 \;\land\; t_2\vDash\lambda_X(g) \implies
  \\ \exists t_3\in FX.\; t_1\le_{FX}t_3\le_{FX}t_2 \;\land\;t_3\vDash\lambda_X(f\land g),
\end{multline}
where $t\vDash\lambda_X(h)$ means that $\lambda_X(h)(t) = \vmaxel$.
In particular, if $\le_{FX}$ is discrete, then this condition can be further simplified, as only the case $t_1 = t_2 = t_3$ needs to be considered.
In this case, $\V$-normality is equivalent to (finitary) \emph{normality}~\cite{Schroder08}.

In the quantale $\zeroinfQ$, $\V$-normality expands as follows: for all $f,g\in [0,\infty]^X$ and all $t_1,t_2\in FX$,
\begin{multline}\label{eqn:v-subadd}
    \inf_{t_3\in FX} d_{FX}(t_1,t_3) + \lambda_X(f+g)(t_3) + d_{FX}(t_3,t_2) \\
    \le \lambda_X(f)(t_1) + d_{FX}(t_1,t_2) + \lambda_X(g)(t_2).
\end{multline}
Similar to before, in the special case where $d_{FX}$ is discrete,~\eqref{eqn:v-subadd} further simplifies to the condition that $d_{FX}$ is \emph{subadditive}.
In the real-valued case we therefore often refer to $\V$-normality as \emph{$\V$-subadditivity} instead.

For general quantales, if every $\V$-category $d_{FX}$ is discrete, then the conditions of Definition~\ref{defn:well-behaved} coincide with or are equivalent to conditions appearing in various related work on Wasserstein constructions~\cite{Hofmann07,BBKK18,BonchiEA23}.

\begin{rem}\label{rem:v-subadd-simpler}
  The condition~\eqref{eqn:v-subadd} in the metric case ($\V=\zeroinfQ$) may still seem relatively complicated, but in actual examples we can typically simplify it a bit further, and show that there exists some $t_3\in FX$ that satisfies both $d_{FX}(t_1,t_3) + d_{FX}(t_3,t_2) = d_{FX}(t_1,t_2)$ and $\lambda_X(f+g)(t_3) \le \lambda_X(f)(t_1) + \lambda_X(g)(t_2)$.
  In practice, $t_3$ is often equal to one of $t_1$ or $t_2$, or easily constructed from them.
\end{rem}

\begin{expl}\label{expl:lists-predicate-lifting}
  The list functor from Example~\ref{expl:lists} admits a well-behaved predicate lifting given by
  \begin{equation*}
    \lambda_X(f)(x_1,\dots x_n) = f(x_1) + \dots + f(x_n).
  \end{equation*}
  This predicate lifting is clearly monotone and preserves the unit zero.
  For $\V$-subadditivity, we follow Remark~\ref{rem:v-subadd-simpler} and let $t_1,t_2\in FX$ and $f,g\colon X\to[0,\infty]$.
  If $t_1$ is a subsequence of $t_2$, put $t_3 \coloneqq t_1$, otherwise put $t_3 \coloneqq t_2$.
  In either case we have $d_{FX}(t_1,t_3) + d_{FX}(t_3,t_2) = d_{FX}(t_1,t_2)$ and $\lambda_X(f+g)(t_3) \le \lambda_X(f)(t_1) + \lambda_X(g)(t_2)$, and~\eqref{eqn:v-subadd} follows.
\end{expl}

The usual Wasserstein extension of a set functor $F\colon\Set\to\Set$ is given in terms of \emph{couplings}, which generalize the concept of distributions with marginals to arbitrary functors.
Given $t_1\in FX$ and $t_2\in FY$, we define the set $\cpl{t_1}{t_2}$ to consist of all those $t\in F(X\times Y)$ such that $F\pi_1(t) = t_1$ and $F\pi_2(t) = t_2$, where $\pi_1\colon X\times Y\to X$ and $\pi_2\colon X\times Y\to Y$ are the projections.
The Wasserstein extension with respect to a predicate lifting $\lambda$ is now defined on a $\V$-valued relation $r\colon\frel{X}{Y}$ via
\begin{equation*}
  \Wc_\lambda r(t_1,t_2) \coloneqq \vJoin \{ \lambda_{X\times Y}(r)(t) \mid t \in \cpl{t_1}{t_2} \},
\end{equation*}
based on the idea that $r$ is a $\V$-valued predicate on $X\times Y$ and can thus be lifted using $\lambda$.
Using diagonals as introduced earlier in this section, we may also represent the Wasserstein extension in pointfree style:
\begin{equation*}
  \Wc_\lambda r = \gr{(F\pi_2)} \cdot \lambdadiag_{X\times Y}(r) \cdot \rev{(F\pi_1)}
\end{equation*}

\begin{expl}\label{expl:prob-wasserstein}
  For $\V=\zeroinfQ$, $F = \dfun$ and $\lambda = \expect{}$, the resulting Wasserstein extension is precisely the probabilistic Kantorovich/Wasserstein extension $\Kant$.
\end{expl}

\begin{expl}\label{expl:simulations-wasserstein-exact}
  Egli-Milner bisimulations for labelled transition systems (Example~\ref{expl:simulations}) can be recovered as $\Wc_\Box$-simulations for the box modality $\Box$ given by
  \begin{equation*}
    U \vDash \Box_X(f) \iff \forall (a,x)\in U.\, x \vDash f.
  \end{equation*}
  Note that this modality ignores the labels, while the back and forth conditions require transitions with matching labels.
  This label-matching is enforced through coupling; in particular, sets with differing sets of labels have no coupling.
\end{expl}

As indicated earlier, the Wasserstein extension of a set functor $F$ is a lax extension when the functor preserves weak pullbacks and the predicate lifting is well-behaved in the above sense, where we view $F$ as a functor into $\VCat$ with each $d_{FX}$ being discrete~\cite{BBKK18,FuzzyLaxHemi}.
For general functors $F\colon\Set\to\VCat$, it is typically not distributorial on its own, so we adapt the definition as follows:

\begin{defn}[Distributorial Wasserstein extension]\label{defn:w}
  Let $F\colon\Set\to\VCat$ be a functor, and let $\lambda$ be a unary predicate lifting for $F$.
  The \emph{distributorial Wasserstein extension}
  $W_\lambda$ is given by
  \begin{equation*}
    W_\lambda r \coloneqq \grv{(F\pi_2)}\cdot\lambdadiag_{X\times Y}(r)\cdot\revv{(F\pi_1)}.
  \end{equation*}
  Given a set $\Lambda$ of predicate liftings, put $W_\Lambda r \coloneqq \vMeet_{\lambda\in\Lambda} W_\lambda r$.
\end{defn}
Alternatively, the distributorial Wasserstein extension can be represented as
$W_\lambda r = d_{FY}\cdot\Wc_\lambda r\cdot d_{FX}$ or 
\begin{align*}
  W_\lambda r
  &= d_{FY}\cdot F\pi_2\cdot\lambdadiag_{X\times Y}(r)\cdot\rev{(F\pi_1)}\cdot d_{FX}.
\end{align*}
Similar to the formulation of the original Wasserstein extension in terms of couplings, we can also give a pointful account of the distributorial Wasserstein extension, which is best understood in the case $\V=\zeroinfQ$:
For $\epsilon \ge 0$, we say that $t\in F(A\times B)$ is an \emph{$\epsilon$-coupling} of $t_1\in FA$ and $t_2\in FB$ if
\begin{equation*}
  d_{FX}(t_1, F\pi_1(t)) + d_{FY}(F\pi_2(t), t_2) \le \epsilon.
\end{equation*}
If we write $\Cpl_\epsilon(t_1,t_2)$ for the set of $\epsilon$-couplings of $t_1$ and $t_2$, then the distributorial Wasserstein extension can equivalently be defined as
\begin{equation*}
  W_\lambda r(t_1,t_2) = \inf \{\epsilon+\lambda_{X\times Y}(r)(t) \mid \epsilon\ge 0, t\in\Cpl_\epsilon(t_1,t_2)\}.
\end{equation*}
This construction gives rise to a distributorial lax extension:
\begin{thm}\label{thm:mw-lax}
  Let $F\colon\Set\to\VCat$ be a functor that preserves exact squares and is cool, and let $\lambda$ be a well-behaved predicate lifting for $F$.
  Then the distributorial Wasserstein extension $W_\lambda$ is a distributorial lax extension.
\end{thm}
\begin{proof}[Proof (sketch)]
  The proofs for items (L1) and (L3) are quite similar to those appearing in previous work, e.g.~\cite[Theorem 6.7]{FuzzyLaxHemi}.
  The proof for item (D) is immediate from the definition of $W_\lambda$ and the fact that both $\revv{(F\pi_1)}$ and $\grv{(F\pi_2)}$ are $\V$-distributors.
  In the proof of (L2), we make use of the fact that the projections $\pi_2\cdot\pi_{12} = \pi_1\cdot\pi_{23}$ from $X\times Y\times Z$ into $Y$ form a pullback and therefore an exact square that is preserved by $F$.
\end{proof}

\begin{expl}\label{expl:simulations-wasserstein}
  Continuing Example~\ref{expl:simulations}, Egli-Milner \emph{simulations} (with only the forth condition) arise by equipping the LTS functor $|F| = \Pow(\A\times(-))$ with the preorder ${\le_{FX}} = {\subseteq}$.
  Details and proofs will be provided in Section~\ref{sec:examples}.
\end{expl}

\begin{expl}\label{expl:lists-wasserstein}
  In Example~\ref{expl:lists}, we considered a hemimetric structure for the list functor.
  The generalized Hamming and Levenshtein distances from this example correspond to the classical and distributorial Wasserstein extension for 
  the well-behaved predicate lifting from Example~\ref{expl:lists-predicate-lifting} that simply sums function values along the list.
  As we are using lax extensions, this distance is even slightly more general than presented earlier, as we can now also work with relations between different alphabets.
  Specifically, we view $\Wc_\lambda r\colon\frel{\List X}{\List Y}$ as the Hamming distance where $r\colon\frel{X}{Y}$ specifies the costs of transforming characters from the alphabet $X$ into characters from the alphabet $Y$, and similarly for the Levenshtein distance.
  The usual Hamming and Levenshtein distance with mutation cost $1$ are recovered for $X=Y$ and $r=\Delta^{01}_X$, where $\Delta^{01}_X(x,y) = 0$ if $x = y$ and $\Delta^{01}_X(x,y) = 1$ otherwise.
\end{expl}

Coolness of the functor at hand is not strictly required, and we can make do without it, but this comes at the cost of requiring a more complicated condition on the predicate lifting.
Specifically:
\begin{prop}\label{prop:mw-lax-not-cool}
  Let $F\colon\Set\to\VCat$ be a functor that preserves exact squares and let $\lambda$ be a well-behaved predicate lifting.
  Then $W_\lambda$ is a distributorial lax extension, provided $\lambda$ satisfies the following additional property:
  for every pair of maps $f\colon X\to Y$ and $g\colon X\to Z$ and every pair of $\V$-valued predicates $p\colon Y\to\V$ and $q\colon Z\to\V$ we have
  \begin{equation*}
    \lambdadiag_Z(q)\cdot\grv{(Fg)}\cdot\revv{(Ff)}\cdot\lambdadiag_Y(p) \\
    \vle \grv{(Fg)}\cdot\lambdadiag_X(p\cdot f\vtimes q\cdot g)\cdot\revv{(Ff)}.
  \end{equation*}
\end{prop}

\begin{rem}\label{rem:barr-extension}
  The earliest approaches to relation lifting~\cite{Barr70} worked by representing a given relation $r\colon\frel{X}{Y}$ as a span $v\cdot\rev{u}$ and then putting $\overline{F} r \coloneqq Fv\cdot\rev{(Fu)}$.
  This is known as the \emph{Barr extension}.
  While the definition does not specify which span should be taken (and is indeed not dependent on this choice), the standard choice for the set at the tip of the span is the relation itself, viewed as a subset of $X\times Y$.
  The maps $u$ and $v$ are then simply the left and right projections.

  We can also view the Barr extension through the lens of the Wasserstein extension, where rather than working with a subset of $X\times Y$ we work with the full set and instead regard $r$ as a predicate on $X\times Y$.
  This predicate can then be lifted using a predicate lifting $\lambda$, resulting in the relation lifting given by $F\pi_2\cdot\delta(\lambda_{X\times Y}(r))\cdot\rev{(F\pi_1)}$.
  The Barr extension is then typically obtained using a box-like modality, see Example~\ref{expl:simulations-wasserstein-exact} and the examples in Sections~\ref{subsec:simulations} and~\ref{subsec:modal}.
  
  Hughes' and Jacobs' coalgebraic treatment of simulations~\cite{HughesJacobs04}, which works by surrounding the Barr extension with preorders on either side, can therefore be seen as an instance of the distributorial Wasserstein extension in the case where the Barr extension is modelled as just described.
\end{rem}

\section{Examples}\label{sec:examples}

\noindent We complement the running examples from Section~\ref{sec:functors} with a number of additional examples, showcasing the versatility of lax couplings.

\subsection{Simulations}\label{subsec:simulations}

\noindent We have already claimed in Example~\ref{expl:simulations-wasserstein} that simulations can be modelled via lax couplings by equipping the LTS functor $\Pow(\A\times(-))$ with the preorder structure of subsethood.
We will now prove this claim and show  also that other types of simulation can be covered in a similar manner by using different preorders.
Specifically, we recover \emph{ready simulation} and \emph{complete simulation}, which are both part of the linear-time/branching-time spectrum~\cite{Glabbeek93}.
For a state $x$ in a labelled transition system, we denote by $I(x)$ the set of \emph{initial actions} of $x$, given by $I(x) = \{a\in\A \mid \exists x'.\, x\overset{a}{\rightarrow}x'\}$~\cite{Glabbeek93}.
\begin{defn}\label{defn:simulations}
  Let $\alpha\colon X\to \Pow(\A\times X)$ and $\beta\colon Y\to \Pow(\A\times Y)$ be labelled transition systems, and let $r\colon\frel{X}{Y}$.
  \begin{enumerate}
    \item We say that $r$ is a \emph{simulation} if for every $x\,r\,y$ and every $(a,x')\in\alpha(x)$ there exists $y'\in Y$ such that $(a,y')\in\beta(y)$ and $x'\,r\,y'$.
    \item We say that $r$ is a \emph{complete simulation} if~$r$ is a simulation and for every $x\,r\,y$ we have $I(x) = \emptyset \iff I(y) = \emptyset$.
    \item We say that $r$ is a \emph{ready simulation} if~$r$ is a simulation and for every $x\,r\,y$ we have $I(x) = I(y)$.
  \end{enumerate}
\end{defn}
\noindent Now consider the following preorders on $\Pow(\A\times X)$, where for $U\in \Pow(\A\times X)$ and $a\in\A$, we put $U_a = \{x\in X \mid (a,x)\in U\}$:
\begin{gather*}
  U \lecomplete V \iff U = V = \emptyset \text{ or } \emptyset \neq U \subseteq V \\
  U \leready V \iff \forall a\in\A.\;U_a = V_a = \emptyset \text{ or } \emptyset \neq U_a \subseteq V_a
\end{gather*}
In addition to the preorders $\subseteq$, $\lecomplete$ and $\leready$, we also consider their converses $\supseteq$, $\geready$ and $\gecomplete$, which correspond to the respective converse simulations.

\begin{lem}\label{lem:simulations-conditions}
  For each choice of
  ${\le_{FX}} \in \{\subseteq, \lecomplete, \leready, \supseteq,
  \gecomplete, \geready\}$, the induced construction
  $F\colon\Set\to\PreOrd$ given by $FX=(\Pow(\A\times
  X),\le_{FX})$ has the following properties:
  \begin{enumerate}
  \item $F$ is functorial, preserves exact squares, and is cool.
    \item $\Box$ is well-behaved as a predicate lifting for~$F$.
  \end{enumerate}
\end{lem}
\noindent By Theorem~\ref{thm:mw-lax}, we thus obtain a distributorial
Wasserstein extension $W^{\le_{FX}}_\Box$ for each choice of
${\le_{FX}} \in \{\subseteq, \lecomplete, \leready, \supseteq,
\gecomplete, \geready\}$. The induced notions of simulation capture
the intended standard concepts:
\begin{lem}\label{lem:characterize-simulations}
  For
  ${\le_{FX}} \in \{\subseteq, \lecomplete, \leready, \supseteq,
  \gecomplete, \geready\}$, $W^{\le_{FX}}_\Box$-simu\-lations are
  characterized as
  \begin{enumerate}
  \item Similarity for ${\le_{FX}}={\subseteq}$, and backwards similarity
    for $\le_{FX}=\supseteq$
  \item (Backwards) complete similarity for ${\le_{FX}}={\lecomplete}$
    (${\le_{FX}}={\gecomplete}$)
  \item (Backwards) ready similarity for ${\le_{FX}}={\leready}$
    (${\le_{FX}}={\geready}$)
  \end{enumerate}
\end{lem}

\subsection{Modal transition systems}\label{subsec:modal}

\noindent Modal transition systems~\cite{LarsenThomsen88} allow
specifying processes using two types of requirements, markings certain
transitions between states as \emph{necessary} or \emph{admissible},
respectively.  Technically, this is modelled through two labelled
transition relations on the same system, typically referred to as
\emph{must} and \emph{may}, where the latter subsumes the former.
These systems come with a dedicated notion of simulation, called
\emph{refinement}, where intuitively one system refines another if its
must and may transitions are more alike; the extreme case, where must
and may coincide, is referred to as an \emph{implementation}. More
recently, an extension of modal transition systems has been
proposed~\cite{BauerEA12} in which the so-called \emph{implementation
  labels} form a partially ordered set $(\A,\vle_\A)$:
\begin{defn}
  A \emph{label-structured modal transition system (LSMTS)} is a tuple $(X,
  \dashedrightarrow,\rightarrow)$, where $X$ is a set of states
  and ${\rightarrow}\subseteq{\dashedrightarrow}\subseteq X\times\A\times X$, where $\rightarrow$ is called the \emph{must} transition relation, and $\dashedrightarrow$ is called the \emph{may} transition relation.
  We write $x\overset{a}{\rightarrow}x'$ for $(x,a,x')\in{\rightarrow}$ and $x\overset{a}{\dashedrightarrow}x'$ for $(x,a,x')\in{\dashedrightarrow}$.
\end{defn}
\noindent Bauer et al.~\cite{BauerEA12} introduce a corresponding
concept of \emph{modal refinement} where mutation of the labels
along~$\vle_\A$ is allowed.  Given LSMTSs with state spaces $X$ and
$Y$, 
a relation
$r\subseteq X\times Y$ is a \emph{modal refinement relation} if
for each $x\,r\,y$ we have:
\begin{itemize}
  \item whenever $x\overset{a}{\dashedrightarrow} x'$, then there exists $y\overset{b}{\dashedrightarrow}y'$ such that $a\vle_\A b$ and $x'\,r\,y'$
  \item whenever $y\overset{b}{\rightarrow} y'$, then there exists $x\overset{a}{\rightarrow} x'$ such that $a\vle_\A b$ and $x'\,r\,y'$
\end{itemize}
LSMSTs and modal refinement can be modelled as an instance of the distributorial Wasserstein extension using the functor $F\colon\Set\to\PreOrd$ given by $|FX| = \{(U,V) \in \Pow(\A\times X)^2 \mid U \subseteq V\}$ and
\begin{align*}
  (U,V)\le_{FX}&(U',V') \iff \\
  & (\forall (a,x)\in V.\;\exists b\in\A.\;a\vle_\A b\land (b,x)\in V') \\
  \land\; & (\forall (b,x)\in U'.\;\exists a\in\A.\;a\vle_\A b\land (a,x)\in U).
\end{align*}
(The original definition includes initial states, which in our setting
just amounts to checking whether specific states are in the modal
refinement relation.) In the modelling, we use the predicate lifting
$\lambda$ given by
\begin{equation*}
  (U,V)\vDash\lambda_X(f) \iff \forall (a,x)\in V.\; x\vDash f,
\end{equation*}
akin to the box modality of the previous subsection.

\begin{lem}\label{lem:modal-ts-conditions}
  $F$ preserves exact squares and is cool and $\lambda$ is well-behaved.
\end{lem}

\begin{lem}\label{lem:characterize-modal-ts}
  $W_\lambda$-simulations coincide with modal refinement relations.
\end{lem}

\subsection{Metric streams}\label{subsec:streams}

\noindent In the real-valued case $\V=\zeroinfQ$, fix a hemimetric space $(\A,d_\A)$ of labels and put $FX = (\A\times X, d_\A\times\Delta_X)$, thus treating~$X$ as a discrete metric space.
$F$-coalgebras correspond to \emph{metric streams}, in the sense that each state in a coalgebra gives rise to an infinite sequence of labels from the hemimetric space $(\A,d_\A)$.
We consider the predicate lifting $\lambda$ given by $\lambda_X(f)(a,x) = f(x)$ for $f\colon X\to\V$, $a\in\A$ and $x\in X$.
Then:
\begin{lem}\label{lem:metric-streams-conditions}
  $F$ preserves exact squares and is cool and $\lambda$ is well-behaved.
\end{lem}

\noindent The corresponding distributorial Wasserstein extension computes the (Manhattan) tensor of $d_\A$ with a given relation, which is given by
\begin{equation*}
  (d_\A\boxplus r)((a,x),(b,y)) = d_\A(a,b) + r(x,y).
\end{equation*}

\begin{lem}\label{lem:characterize-metric-streams}
  For every fuzzy relation $r$, $W_\lambda r = d_\A \boxplus r$.
\end{lem}
\begin{proof}
  Expanding definitions and disregarding terms evaluating to $\infty$,
  \begin{align*}
    W_\lambda r((a,x),(b,y)) &= \inf_{c\in\A} d_\A(a,c) + r(x,y) + d_\A(c,b) \\
    &= d_\A(a,b) + r(x,y). \qedhere
  \end{align*}
\end{proof}

\subsection{Monoid-valued functors}\label{subsec:monoids}

\noindent Fix a monoid $(M,+,0)$ equipped with a hemimetric $d_M$ and
consider the \emph{monoid-valued functor} $|F| = M^{(-)}$, which maps
a set $X$ to the set $M^{(X)}$ of \emph{finitely supported} maps
$\mu\colon X\to M$, i.e.\ $\mu(x) = 0$ for all but finitely many
$x\in X$.  On functions $f\colon X\to Y$, $M^{(-)}$ acts as
$M^{(f)}(\mu)(y) = \sum_{f(x)=y} \mu(x)$.

We extend $M^{(-)}$ to a functor $F\colon\Set\to\HMet$ by equipping
each set $FX$ with the \emph{total variation distance} $\dtv$, where
\begin{equation*}
  \dtv(\mu,\nu) = \sup_{U\subseteq X} d_M(\mu(U),\nu(U)).
\end{equation*}
We consider two specific instances of this construction: the \emph{bag functor} $\Bag$, which corresponds to the monoid $M = \Nat$ of natural numbers, and the \emph{finite measure functor}, which corresponds to the monoid $M = \Real^+$ of nonnegative real numbers.
In both cases we use addition as the monoid structure and \emph{truncated subtraction} $d_M(x,y) = \max(0, y-x)$ as the hemimetric.
One immediate consequence is that in both cases, total variation distance simplifies to
\begin{equation*}
  \dtv(\mu,\nu) = \textstyle\sum_{x\in X} d_M(\mu(x),\nu(x)).
\end{equation*}

\begin{lem}\label{lem:distrib-exact}
  For both choices of monoid above, the functor $F\colon\Set\to\HMet$ preserves exact squares and is cool.
\end{lem}

As predicate lifting we use the \emph{expected value modality} $\expect{}$, which is defined by
\begin{equation}\label{eqn:dtv-simpler}
  \expect{X}(\mu)(f) = \textstyle\sum_{x\in X} \mu(x)\cdot f(x).
\end{equation}

\begin{lem}\label{lem:expect-well-behaved}
  The predicate lifting $\expect{}$ is well-behaved.
\end{lem}
\begin{proof}
  $\expect{}$ is clearly monotonic and preserves the zero function; we show $\V$-subadditivity.
  Let $X$ be a set, let $\mu,\nu\in FX$ and let $f,g\colon X\to[0,\infty]$.
  We follow the recipe laid out in Remark~\ref{rem:v-subadd-simpler} and define $\rho\in FX$ by $\rho(x) = \min(\mu(x),\nu(x))$ for each $x\in X$.
  According to~\eqref{eqn:dtv-simpler}, we have $\dtv(\mu,\nu) = \dtv(\mu,\rho) + \dtv(\rho,\nu)$ and by definition of $\rho$ and linearity of expectation we also have
  \begin{equation*}
    \expect{\rho}(f+g) = \expect{\rho}(f) + \expect{\rho}(g) \le \expect{\mu}(f) + \expect{\nu}(g). \qedhere
  \end{equation*}
\end{proof}

The usual Wasserstein metric, which arises from the lax extension
$\Wc_{\expect{}}$, measures distances between probability measures, or
more generally between measures of equal total mass. From
Lemmas~\ref{lem:distrib-exact} and~\ref{lem:expect-well-behaved}, we
obtain a \emph{distributorial Wasserstein extension $W_{\expect{}}$
  that assigns non-trivial distances also to measures of different
  total mass}.

Already Kantorovich~\cite{Kantorovich39} observed that the computation
of the distance between two measures can be viewed as the optimization
of a \emph{transportation problem} of the mass from the first measure
to the second. In recent years, the problem of defining distances
between measures of potentially non-equal total mass has received a
considerable amount of attention in the artificial intelligence and
data science communities under the name of \emph{unbalanced optimal
  transport}~\cite{LieroEA18,ChizatEA18a,ChizatEA18b}.  In this
context, such distances are used as loss functions between weighted
samples as they arise in applications such as image
processing~\cite{RubnerEA97,LeeBR20} or the analysis of machine
learning methods~\cite{ChizatB18}.  Piccoli and
Rossi~\cite{PiccoliRossi14,PiccoliRossi16} define a generalized
Wasserstein distance for Borel measures which, applied to finite
measures, exactly coincides with our above-mentioned distributorial
Wasserstein extension $W_{\expect{}}$.  More recent
approaches~\cite{LieroEA18} replace the total variation distance by
more general penalizing functionals and cannot be covered by the
current framework, which requires at least a hemimetric structure.
Capturing these more general distances in coalgebraic generality is an
interesting direction for future work.

\subsection{Metric labelled Markov chains}\label{subsec:mlmc}

\noindent Continuing  Example~\ref{expl:mlmc}, we work with the functor $FX = (\dfun(\A\times X), \Kant(d_\A\times \Delta_X))$ that applies the Kantorovich/Wasserstein lifting to probability measures on $\A\times X$ while treating $X$ as a discrete metric space.
Viewed in terms of transportation plans (see the discussion in the previous subsection), $d_{FX}$ measures the distance between such probability measures by only considering plans where goods are transported along the $\A$ dimension, with probability mass remaining constant inside each slice.

To combine the distance on $FX$ with a given distance on a set $X$, or more generally with a fuzzy relation $r\colon\frel{X}{Y}$ modelling the distances from set $X$ to set $Y$, we consider the distributorial Wasserstein extension for the predicate lifting $\lambda$ given by
\begin{equation*}
  \lambda_X(f)(\mu)
  \coloneqq \expect{\dfun\pi_2(\mu)}(f)
  = \;\;\sum_{\mathclap{(a,x)\in\A\times X}}\;\; \mu(a,x)\cdot f(x).
\end{equation*}

\begin{lem}\label{lem:mlmc-conditions}
  $F$ preserves exact squares and is cool,
  and $\lambda$ is well-behaved.
\end{lem}

Following the common theme from previous examples, this predicate lifting essentially ignores the labelling information present in the left component $\A$ and instead only cares about the distance values that are supplied when applying $\lambda$ to the lifted relation $r$.
Consequently, we obtain the following characterization for the corresponding classical Wasserstein extension:

\begin{lem}\label{lem:mlmc-representation-classical}
  We have $\Wc_\lambda r = \Kant(\Delta_\A\times r)$ for every fuzzy relation $r\colon\frel{X}{Y}$.
\end{lem}

\noindent
To obtain a characterization for the corresponding distributorial Wasserstein extension $W_\lambda$, we make use of the following decomposition that holds whenever the Wasserstein extension is applied to the tensor (as seen in Section~\ref{subsec:streams}) of two distances:

\begin{lem}\label{lem:mlmc-decompose}
  For every fuzzy relation $r\colon\frel{X}{Y}$,
  \begin{multline*}  
    \Kant(d_\A\boxplus r) = \\
    \Kant(d_\A\times\Delta_Y) \cdot \Kant(\Delta_\A\times r) \cdot \Kant(d_\A\times\Delta_X).
  \end{multline*}
\end{lem}

\begin{cor}\label{cor:mlmc-representation}
  We have $W_\lambda r = \Kant(d_\A\boxplus r)$ for every fuzzy relation $r\colon\frel{X}{Y}$.
\end{cor}
\noindent That is, while the classical Wasserstein
extension~$\Wc_\lambda$ ignores the metric on the label set~$\A$,
replacing it with the discrete metric, the distributorial Wasserstein
extension $W_\lambda$ takes the metric on~$\A$ into account in the
intended manner.
\begin{proof}[Proof (Corollary~\ref{cor:mlmc-representation})]
  As $\Kant(d_\A\times\Delta_X) = d_{FX}$ and $\Kant(\Delta_\A\times r) = F\pi_2\cdot\delta(\lambda_{X\times Y})\cdot\rev{(F\pi_1)}$, this is immediate from Lemma~\ref{lem:mlmc-decompose}.
\end{proof}

\begin{rem}\label{rem:optimized-flow}
  The decomposition from Lemma~\ref{lem:mlmc-decompose} is not only helpful for characterizing the distributorial Wasserstein extension, but may also be viewed as a recipe to optimize its computation.
  Specifically, if $\A$, $X$ and $Y$ are finite sets, for simplicity of the same size $n$, then the distance between two probability measures on $\A\times X$ and $\A\times Y$ can be computed using a network flow algorithm on a bipartite graph with $\A\times X$ and $\A\times Y$ as the left and right partition, respectively~\cite{BreugelWorrell06}.
  As this graph has $\Theta(n^4)$ edges, the distance between two probability measures can be computed in $\mathcal{O}(n^6\cdot\log^2(n))$ time using the network simplex algorithm~\cite{Tarjan97}.

  Using Lemma~\ref{lem:mlmc-decompose}, we can instead model the computation of the distance using a tiered network with four layers, where the first two layers are copies of $\A\times X$ and the second two layers are copies of $\A\times Y$.
  The edges between these layers correspond to the relations $d_\A\times\Delta_X$, $\Delta_\A\times r$ and $d_\A\times\Delta_Y$, which together only have $\mathcal{O}(n^3)$ edges with finite cost.
  This results in an improved runtime complexity of $\mathcal{O}(n^5\cdot\log^2(n))$ when using the same flow algorithm as before.
\end{rem}

\section{Duality}\label{sec:duality}

\noindent Already in the early work on coalgebraic behavioural distances~\cite{BBKK14}, the question was raised whether the generalizations of the two sides of the Kantorovich/Rubinstein duality admit a duality of their own, that is, whether $K_\lambda = \Wc_\lambda$ for a given choice of functor and predicate lifting
(in the sense that $K_\lambda r = \Wc_\lambda r$ for every relation $r$).
Duality in this strict sense is only known to hold in very few cases; other than the original probabilistic case, the most well-known instance is that of the pseudometric Hausdorff lifting of the powerset functor~\cite{BBKK18}, which is modelled using the predicate lifting $\sup$, where $\sup_X(f)(U) = \sup_{x\in U} f(x)$.
In the case of lax extensions, where the Kantorovich extension is treated in terms of not necessarily symmetric relations, this equality breaks, but can be recovered by adding the dual predicate lifting $\inf$ on the Kantorovich side: $K_{\{\sup,\inf\}} = \Wc_{\sup}$.
This motivates the question of finding more generalized dualities of the form $K_\Lambda = \Wc_\lambda$, where $\Lambda$ is a set of predicate liftings.

In a certain sense, this question has already been settled through a representation theorem~\cite{FuzzyLaxHemi} which shows that, under mild assumptions on the functor, every lax extension (and therefore in particular every Wasserstein extension) can be expressed as a Kantorovich extension for a suitable set of predicate liftings.
The sets of predicate liftings arising from this construction, however, are often quite large, being defined in terms of a representation of the underlying functor.
More recently, a more explicit construction for polynomial variations of the powerset and discrete distribution functors has been proposed~\cite{HumeauEA25}.

We now discuss some examples of dualities in the distributorial setting, i.e.~dualities of the shape $K_\Lambda = W_\lambda$, where $\lambda$ is a predicate lifting and $\Lambda$ is a set of predicate liftings.

\begin{expl}\label{expl:sim-duality}
  Continuing the discussion on Egli-Milner simulation (Example~\ref{expl:simulations-wasserstein}), we consider the diamond modality $\Diamond$, where
  \begin{equation*}
    U \vDash \Diamond_X(f) \iff \exists (a,x)\in U.\; x \vDash f.
  \end{equation*}
  This predicate lifting $\Diamond$ satisfies the condition from Theorem~\ref{thm:k-lax}, and indeed we also have the equality $K_\Diamond = W_\Box$.
  For converse Egli-Milner simulation (where the preorder on the functor is $\supseteq$ instead of $\subseteq$), one can check that now it is instead the box modality $\Box$ that satisfies the condition from Theorem~\ref{thm:k-lax} and we have the equality $K_\Box = W_\Box$.
  Finally, Egli-Milner bisimulation, which arises by equipping the sets $FX$ with the discrete preorder, is induced by the Kantorovich extension for $\{\Box,\Diamond\}$.
\end{expl}

\begin{expl}\label{expl:stream-duality}
  The distributorial Wasserstein extension for metric streams (Section~\ref{subsec:streams}) admits a Kantorovich representation which is given by the predicate liftings $\Diamonda{a}$, where $a$ ranges over $\A$ and
  \begin{equation*}
    \Diamonda{c}_X(f)(a,x) \coloneqq d_\A(a,c) + f(x).
  \end{equation*}
  Explicitly, we have $K_\Lambda = W_\lambda$, where $\Lambda \coloneqq \{\Diamonda{a} \mid a\in\A\}$ and $\lambda$ is as in Section~\ref{subsec:streams}.
\end{expl}

\section{Conclusion and future work}\label{sec:conclusion}

\noindent
We have defined a generalization of the coupling-based Wasserstein extension that applies to functors from sets to categories of conformances (either preorders or hemimetrics), where the given conformance structure enables us to consider lax couplings of functor elements, thus bypassing the issue that exact couplings need not exist in general.
The key technical aspect of this construction is the usage of $\V$-distributors, whose properties ensure that one obtains a lax extension and is therefore able to derive behavioural conformances, functor liftings and notions of (bi)simulation.
We have demonstrated the versatility of this construction by showing that it generalizes various known concepts, such as simulations for both labelled and modal transition systems and a generalized Wasserstein metric that applies to finite measures of possibly different masses, as well as natural notions of distance for metric streams and metric-labelled Markov chains.

The latter setting of unbalanced optimal transport motivates further research, as many of the constructions in this context are based around weak notions of distance such as diffuse metrics or divergences, which do not fit within the framework of $\V$-categories and thus are also not covered by distributorial Wasserstein extensions.
A related issue revolves around the question of replacing the Manhattan tensor occurring in the distances for metric-labelled systems with the more commonly used product metric, which similarly calls for a yet more general notion of lax coupling that is detached from the distributorial setting and allows distance values to be combined in different ways than just through summation.

\bibliographystyle{splncs04}
\bibliography{fw}

\onecolumn
\appendix

\section*{Details for Example~\ref{expl:simulations}}

\noindent
We prove that for every pair of sets $U\subseteq\A\times X$ and $V\subseteq\A\times Y$ we have $U\,L^\subseteq r\,V$ if and only if for $(a,x)\in U$ there exists $(a,y)\in V$ such that $x\,r\,y$.
\begin{itemize}
  \item ``$\Rightarrow$'':
    Let $U\,L^\subseteq r\,V$.
    Then there exist sets $U'$ and $V'$ such that $U\subseteq U'$, $V'\subseteq V$ and $U'\,Lr\,V'$.
    Let $(a,x)\in U$.
    Then we also have $(a,x)\in U'$, and thus by definition of $L$ there exists $(a,y)\in V'$ such that $x\,r\,y$.
    As $V'\subseteq V$, we also have $(a,y)\in V$.
  \item ``$\Leftarrow$'':
    Now assume that $U$ and $V$ satisfy the condition that for every $(a,x)\in U$ there exists $(a,y)\in V$ such that $x\,r\,y$.
    We define $U'\coloneqq\{(a,x)\in\A\times X \mid \exists (a,y)\in V.\; x\,r\,y\}$ and
    $V'\coloneqq \{(a,y)\in V \mid \exists (a,x)\in U'.\; x\,r\,y\}$.
    We now need to show that $U\subseteq U'$, $U'\,Lr\,V'$ and $V'\subseteq V$.
    The third of these claims holds by definition.
    For the first, let $(a,x)\in U$.
    Then by assumption there exists $(a,y)\in V$ such that $x\,r\,y$, so $(a,x)\in U'$ by definition of $U'$.
    Finally, we discharge the second claim by showing the forth and back conditions:
    \begin{itemize}
      \item Let $(a,x)\in U'$.
        Then, by definition of $U'$, there exists $(a,y)\in V$ such that $x\,r\,y$.
        Note that this also entails $(a,y)\in V'$, which means there is nothing left to prove here.
      \item Let $(a,y)\in V'$. Then by assumption there exists $(a,x)\in U'$ such that $x\,r\,y$, which finishes the proof.
    \end{itemize}
\end{itemize}

\section*{Proof of Lemma~\ref{lem:distrib-lax}}

\begin{itemize}
  \item (D) $\to$ (D'): By assumption and using (L2),
    \begin{equation*}
      d_{FB}\cdot Lr \vle L\Delta_B\cdot Lr \vle L(\Delta_B\cdot r) = Lr,
    \end{equation*}
    and the inequality $Lr\cdot d_{FA}\vle Lr$ follows similarly.
  \item (D') $\to$ (D): Using (L3) and the assumption,
    \begin{equation*}
      d_{FX} = d_{FX}\cdot\Delta_{FX} \vle d_{FX}\cdot L\Delta_X \vle L\Delta_X. \qedhere
    \end{equation*}
\end{itemize}

\section*{Proof of Lemma~\ref{lem:exact-via-pullback}}

\noindent
`Only if' is clear, as every pullback square in $\Set$ is exact; we show `if'.
Consider a weak pullback square~\eqref{eqn:diag-pxyz}, and let $(Q, s\colon Q\to X, t\colon Q\to Y)$ be the pullback of $f$ and $g$.
Let $c\colon P\to Q$ be the unique map satisfying $u=s\cdot c$ and $v=t\cdot c$ which exists by the universal property of $Q$, and let $d\colon Q\to P$ be some map satisfying $s=u\cdot d$ and $t=v\cdot d$ that exists by the weak universal property of $P$.
Then $c\cdot d = \id_Q$ by the universal property of $Q$, so $c$ is a split epi.
Applying $F$, we get $Fc\cdot Fd = \id_{FQ}$, so $Fc$ is also a split epi.
$Fc$ is therefore surjective when viewed as a set map and thus satisfies $\Delta_{FQ} \vle Fc\cdot\rev{(Fc)}$.
We use this to show that the exact square~\eqref{eqn:diag-pxyz} is preserved:
\begin{align*}
  \rev{(Fg)}\cdot d_{FZ}\cdot Ff
  &\vle d_{FY}\cdot Ft\cdot\rev{(Fs)}\cdot d_{FX} && \by{assumption} \\
  &\vle d_{FY}\cdot Ft\cdot Fc\cdot\rev{(Fc)}\cdot\rev{(Fs)}\cdot d_{FX} && \by{$\Delta_{FQ} \vle Fc\cdot\rev{(Fc)}$} \\
  &= d_{FY}\cdot Fv\cdot\rev{(Fu)}\cdot d_{FX} && \by{$v=t\cdot c$ and $u=s\cdot c$}
\end{align*}

\section*{Proof of Lemma~\ref{lem:cool-sym}}

\noindent
If every $d_{FX}$ is symmetric, then $\rev{(Ff)}\cdot d_{FY} \vle d_{FX}\cdot\rev{(Ff)} \iff \rev{(Ff)}\cdot \rev{d_{FY}} \vle \rev{d_{FX}}\cdot\rev{(Ff)} \iff d_{FY}\cdot Ff \vle Ff\cdot d_{FX}$.

\section*{Details for Example~\ref{expl:running-cool}}

\noindent
For the metric-labelled Markov chain functor $FX = (\dfun(\A\times X), \Kant(d_\A\times\Delta_X))$, consider $\mu\in FX$ and $\nu'\in FY$, put $\nu = Ff(\mu)$ and put $\epsilon = d_{FY}(\nu,\nu')$.
If $\epsilon = \infty$ we are done right away, so assume $\epsilon < \infty$.
By general results from transportation theory~\cite{Villani09}, there exists an optimal coupling $\sigma$ between $\nu$ and $\nu'$.
Because $\epsilon < \infty$, we know that $\sigma(a,y,b,y') = 0$ whenever $y \neq y'$.
We now construct a probability distribution $\mu'\in FX$ and a coupling $\rho$ of $\mu$ and $\mu'$ simultaneously using the following algorithm:
\begin{enumerate}
  \item Initialize $\mu'$ and $\rho$ with zeroes.
  \item For every $(a,y,b,y)\in(\A\times Y)^2$ such that $\sigma(a,y,b,y) > 0$ and every $x\in f^{-1}(x)$, let $\delta = \frac{\mu(a,x)}{\nu(a,y)}\cdot \sigma(a,y,b,y)$ and increase both $\mu'(b,x)$ and $\rho(a,x,b,x)$ by $\delta$.
\end{enumerate}
One can then verify that indeed $\rho$ is a coupling of $\mu$ and $\mu'$ and that $d_{FY}(\nu,\nu') = \expect{\sigma}(d_\A\times\Delta_Y) \ge \expect{\rho}(d_\A\times\Delta_X) = d_{FX}(\mu,\mu')$.
By Lemma~\ref{lem:cool-sym} we do not need to show the other inequality.

\section*{Proof of Theorem~\ref{thm:mw-lax}}

\noindent
We first note the following fact about diagonal relations:

\begin{lem}\label{lem:g-diag-prop}
  Let $f\colon X\to Y$ and $g\colon Y\to\V$.
  Then $f\cdot\diag{g\cdot f} = \diag{g} \cdot f$.
\end{lem}
\begin{proof}
  Let $x\in X$ and $y\in Y$, and evaluate the relations at $(x,y)$.
  Then both sides of the equation are equal to $g(f(x))$ if $f(x) = y$, and $\vminel$ otherwise.
\end{proof}

\noindent We show the axioms of distributorial lax extensions one by one:
\begin{description}
  \item[(L1)] Immediate by monotonicity of $\lambda$ and of relational composition.
  \item[(L2)] Let $r\colon\frel{X}{Y}$ and $s\colon\frel{Y}{Z}$.
    We expand $W_\lambda s\cdot W_\lambda r$, but for space reasons omit the first two and last two $\V$-relations in the chain for now.
    For the remaining middle part we proceed as follows:
    \begin{align*}
      &\lambdadiag_{Y\times Z}(s)\cdot
          \rev{(F\pi_1)}\cdot d_{FY}\cdot d_{FY}\cdot F\pi_2\cdot
          \lambdadiag_{X\times Y}(r) \\
      &= \lambdadiag_{Y\times Z}(s)\cdot
          \rev{(F\pi_1)}\cdot d_{FY}\cdot F\pi_2\cdot
          \lambdadiag_{X\times Y}(r) \\
      \intertext{
        using that $d_{FY}$ is a $\V$-category.
        Now, observe that the projections from $X\times Y\times Z$ into $Y$ form a pullback, and hence an exact square $\rev{\pi_1}\cdot\pi_2 = \pi_{23}\cdot\rev{\pi_{12}}$, which is preserved by $F$:
      }
      &\vle \lambdadiag_{Y\times Z}(s)\cdot
          d_{F(Y\times Z)}\cdot
          F\pi_{23}\cdot\rev{(F\pi_{12})}\cdot
          d_{F(X\times Y)}\cdot
          \lambdadiag_{X\times Y}(r) \\
      \intertext{
        Now we may use coolness of $F$ to rearrange in the middle:
      }
      &\vle \lambdadiag_{Y\times Z}(s)\cdot
          F\pi_{23}\cdot
          d_{F(X\times Y\times Z)}\cdot
          d_{F(X\times Y\times Z)}\cdot
          \rev{(F\pi_{12})}\cdot
          \lambdadiag_{X\times Y}(r) \\
      &\vle \lambdadiag_{Y\times Z}(s)\cdot
          F\pi_{23}\cdot
          d_{F(X\times Y\times Z)}\cdot
          \rev{(F\pi_{12})}\cdot
          \lambdadiag_{X\times Y}(r) \\
      \intertext{
        Next, as $\lambda$ is natural, the diagonal relations can be moved to the inside.
        We define $r_{12}, s_{23}\colon X\times Y\times Z\to\V$ by $r_{12}(x,y,z) = r(x,y)$ and $s_{23}(x,y,z) = s(y,z)$ respectively.
        Then $\lambda_{X\times Y}(r)\cdot F\pi_{12} = \lambda_{X\times Y\times Z}(r_{12})$ and $\lambda_{Y\times Z}(s)\cdot F\pi_{23} = \lambda_{X\times Y\times Z}(s_{23})$ by naturality of $\lambda$. 
        Therefore, using Lemma~\ref{lem:g-diag-prop},
      }
      &= F\pi_{23}\cdot\lambdadiag_{X\times Y\times Z}(s_{23})
          \cdot d_{F(X\times Y\times Z)}\cdot
          \lambdadiag_{X\times Y\times Z}(r_{12})\cdot\rev{(F\pi_{12})} \\
      \intertext{
        and we may now use $\V$-normality of $\lambda$ to obtain:
      }
      &\vle F\pi_{23}\cdot d_{F(X\times Y\times Z)}\cdot
          \lambdadiag_{X\times Y\times Z}(r_{12}\vtimes s_{23})
          \cdot d_{F(X\times Y\times Z)}\cdot\rev{(F\pi_{12})}.
    \end{align*}
    Next, we combine the left two terms of this with the left two terms we omitted at the start and get:
    \begin{align*}
      &d_{FZ}\cdot F\pi_2\cdot F\pi_{23}\cdot d_{F(X\times Y\times Z)} \\
      &= d_{FZ}\cdot F(\pi_2\cdot \pi_{23})\cdot d_{F(X\times Y\times Z)}
      &&\by{$F$ is a functor} \\
      &\vle d_{FZ}\cdot d_{FZ}\cdot F(\pi_2\cdot \pi_{23})
      &&\by{$F(\pi_2\cdot F\pi_{23})$ is a $\V$-functor} \\
      &= d_{FZ}\cdot F(\pi_2\cdot \pi_{23})
      &&\by{$d_{FZ}$ is a $\V$-category} \\
      &= d_{FZ}\cdot F(\pi_2\cdot \pi_{13})
      &&\by{$\pi_2\cdot \pi_{23} = \pi_2\cdot \pi_{13}$} \\
      &= d_{FZ}\cdot F\pi_2\cdot F\pi_{13}
      &&\by{$F$ is a functor}
    \end{align*}
    Proceeding similarly with the respective terms on the right yields
    \begin{equation*}
      d_{F(X\times Y\times Z)}\cdot\rev{(F\pi_{12})}\cdot\rev{(F\pi_1)}\cdot d_{FX} \vle \rev{(F\pi_{13})}\cdot\rev{(F\pi_1)}\cdot d_{FX}.
    \end{equation*}
    For the middle term, we observe that for all $(x,y,z)\in X\times Y\times Z$ we have $(r_{12}\vtimes s_{23})(x,y,z) = r(x,y)\vtimes s(y,z) \vle (s\cdot r)(x,z)$.
    Therefore, by monotonicity and naturality of $\lambda$,
    \begin{equation*}
      \lambdadiag_{X\times Y\times Z}(r_{12}\vtimes s_{23}) \vle
      \diag{\lambda_{A\times C}(s\cdot r)\cdot F\pi_{13}}.
    \end{equation*}
    Gathering all terms together, we get:
    \begin{align*}
      &W_\lambda(s)\cdot W_\lambda(r) \\
      &\vle d_{FZ}\cdot F\pi_2\cdot F\pi_{13}\cdot
          \diag{\lambda_{X\times Z}(s\cdot r)\cdot F\pi_{13}}\cdot
          \rev{(F\pi_{13})}\cdot\rev{(F\pi_1)}\cdot d_{FX} \\
      \intertext{
        Using Lemma~\ref{lem:g-diag-prop} once more, as well as the fact that $f\cdot\rev{f} \vle\Delta_Y$ for each $f\colon X\to Y$, we finish up the proof:
      }
      &= d_{FZ}\cdot F\pi_2\cdot
          \diag{\lambda_{X\times Z}(s\cdot r)}\cdot
          F\pi_{13}\cdot\rev{(F\pi_{13})}\cdot\rev{(F\pi_1)}\cdot d_{FX} \\
      &\vle d_{FZ}\cdot F\pi_2\cdot
          \diag{\lambda_{X\times Z}(s\cdot r)}\cdot
          \rev{(F\pi_1)}\cdot d_{FX} = W_\lambda(s\cdot r).
    \end{align*}
  \item[(L3)] Let $f\colon A\to B$ and consider its graph $\gr{f}\colon A\times B\to\V$.
    Then, as $\lambda$ preserves the unit and is natural,
    \begin{equation}\label{eqn:mw-lax-3}
      \vunit_{FA} \vle \lambda_A(\vunit_A) = \lambda_A(\gr{f}\cdot\langle\id_A,f\rangle) = \lambda_{A\times B}(\gr{f})\cdot F\langle\id_A,f\rangle
    \end{equation}
    We also observe that $F\langle\id_A,f\rangle\vle\rev{(F\pi_1)}$, and thus
    \begin{align*}
      &F(\gr{f}) \\
      &= F\pi_2\cdot F\langle\id_A,f\rangle\cdot\diag{\vunit_{FA}} \\
      &\vle F\pi_2\cdot F\langle\id_A,f\rangle\cdot\diag{\lambda_{A\times B}(\gr{f})\cdot F\langle\id_A,f\rangle} &&\by{\ref{eqn:mw-lax-3}} \\
      &\vle F\pi_2\cdot\diag{\lambda_{A\times B}(\gr{f})}\cdot F\langle\id_A,f\rangle &&\by{Lemma~\ref{lem:g-diag-prop}} \\
      &\vle F\pi_2\cdot\diag{\lambda_{A\times B}(\gr{f})}\cdot \rev{(F\pi_1)} &&\by{fact}\\
      &\vle d_{FB}\cdot F\pi_2\cdot\diag{\lambda_{A\times B}(\gr{f})}\cdot \rev{(F\pi_1)}\cdot d_{FA} \\
      &= W_\lambda(\gr{f})
    \end{align*}
    using reflexivity of $\V$-categories in the last step.
    The proof of the other inequality is analogous.
  \item[(D)] By Lemma~\ref{lem:distrib-lax}, we need to show that $d_{FY}\cdot W_\lambda r\cdot d_{FX} \vle W_\lambda r$ for each $r$.
    The claim then follows from the triangle inequalities $d_{FX}\cdot d_{FX} \vle d_{FX}$ and $d_{FY}\cdot d_{FY} \vle d_{FY}$. \qedhere
\end{description}

\section*{Proof of Proposition~\ref{prop:mw-lax-not-cool}}

\noindent
We only need to modify one part of the proof of (L2) in the proof of Theorem~\ref{thm:mw-lax}.
Instead of applying coolness of the functor $F$, followed by naturality and $\V$-normality of $\lambda$, we directly apply the new property of the predicate lifting and skip these steps.
The rest of the proof can then proceed just as before.

\section*{Proof of Lemma~\ref{lem:simulations-conditions}}

\noindent
First, note the following two simple facts about $F$:

\begin{lem}\label{lem:lts-fun-preserves-empty}
  Let $f\colon X\to Y$, and $U\in FX$. Then
  \begin{enumerate}
    \item $U = \emptyset \iff Ff(U) = \emptyset$
    \item For each $a\in\A$, $U_a = \emptyset \iff (Ff(U))_a = \emptyset$.
  \end{enumerate}
\end{lem}

\begin{enumerate}
  \item It is straightforward to see that each $FX$ is a preorder and each $Ff$ is monotone, so $F$ is indeed a functor.
    For preservation of exact squares, let $f\cdot u = g\cdot v$ be an exact square (where $u\colon P\to X$, $v\colon P\to Y$, $f\colon X\to Z$ and $g\colon Y\to Z$), let $U\in FX$ and $V\in FY$ such that $Ff(U) \le_{FZ} Fg(V)$.
    We need to find $W\in FP$ such that $U \le_{FX} Fu(W)$ and $Fv(W) \le_{FY} V$.
    Such a set $W$ can be constructed by selecting, for each $(a,x)\in U$ and $(a,y)\in V$ such that $f(x) = g(y)$, some $p\in P$ such that $u(p) = x$ and $v(p) = y$ (which exists by exactness of $f\cdot u = g\cdot v$), and adding $(a,p)$ to $W$.
    We show the two inequalities in the case ${\le_{FX}} = {\subseteq}$:
    \begin{itemize}
      \item Let $(a,x) \in U$.
        Then, by assumption, there exists $y\in Y$ such that $(a,y)\in V$ and $f(x) = g(y)$.
        By construction of $W$, then, there exists some $(a,p)\in W$ such that $u(p) = x$, which implies $(a,x)\in Fu(W)$.
      \item Let $(a,p) \in W$.
        Then, by construction of $W$, $v(p) = y$ for some $y\in V_a$, implying $Fv(W) \subseteq V$.
    \end{itemize}
    The other cases can be shown similarly, additionally using Lemma~\ref{lem:lts-fun-preserves-empty} in the non-Egli-Milner cases.
  \item Monotonicity is easy to see, as is preservation of $\vmaxel$ -- we show $\V$-normality.
    Let $f,g\in\V^X$ and $U,V\in FX$ such that $U \le_{FX} V$ and $\Box_X(f)(U) = \Box_X(g)(V) = \vmaxel$.
    We need to find $W\in FX$ such that $\Box_X(f\land g)(W) = \vmaxel$ and $U \le_{FX} W \le_{FX} V$.
    Put $W \coloneqq U \cap V$.
    Then $\Box_X(f)(W) \vge \Box_X(f)(U) = \vmaxel$ and $\Box_X(g)(W) \vge \Box_X(g)(V) = \vmaxel$ by definition of $\Box$, and thus also $\Box_X(f\land g)(W) = \vmaxel$.
    In all six cases, we have $W = U$ or $W = V$, so that the other condition also holds.
  \qedhere
\end{enumerate}

\section*{Proof of Lemma~\ref{lem:characterize-simulations}}

\noindent
Let $\alpha\colon X\to FX$ and $\beta\colon Y\to FY$ be labelled transition systems, and let $r\colon\frel{X}{Y}$.
We prove the three cases one by one.
\begin{description}[wide]
  \item[`$\subseteq$':] First, assume that $r$ is a $W_\Box$-simulation.
    Let $(x,y) \in r$, and let $(a,x')\in\alpha(x)$.
    By assumption, $(\alpha(x),\beta(y)) \in W_\Box r$, which means that there exists a set $U \subseteq \A\times X\times Y$ such that $\alpha(x) \subseteq F\pi_1(U)$ and $F\pi_2(U)\subseteq \beta(y)$ and for all $(a',x'',y'') \in U$ we have $(x'',y'')\in r$.
    As $\alpha(x) \subseteq F\pi_1(U)$, we have $(a,x)\in F\pi_1(U)$, meaning that there exists $y'\in Y$ such that $(a,x',y') \in U$.
    Therefore, $(x',y') \in r$, and also $(a,y') \in F\pi_2(U) \subseteq \beta(y)$, which is exactly what we needed to show.

    Second, assume that $r$ is an Egli-Milner simulation, and let $(x,y) \in r$.
    We need to show that $(\alpha(x),\beta(y))\in W_\Box r$.
    By assumption, for each $(a,x')\in\alpha(x)$ there is $y'\in Y$ such that $(a,y')\in\beta(y)$ and $(x',y')\in r$, so construct $U\subseteq\A\times X\times Y$ by taking all the triples $(a,x',y')$ arising in this way.
    By construction, we have $\alpha(x) \subseteq F\pi_1(U)$ and $F\pi_2(U)\subseteq \beta(y)$, as well as $\Box_{X\times Y}(r)(U) = \vmaxel$, so $U$ witnesses that $(\alpha(x),\beta(y))\in W_\Box r$.
  \item[`$\lecomplete$':] First, assume that $r$ is a $W_\Box$-simulation.
    Let $(x,y) \in r$. We proceed as in the previous item, just with $\lecomplete$ in place of $\subseteq$.
    In particular, if $\alpha(x)\neq\emptyset$, then both $I(x)\neq\emptyset$ and $I(y)\neq\emptyset$.
    Otherwise, if $\alpha(x)=\emptyset$, then both $U$ and $\beta(y)$ must be empty as well, so $I(x) = I(y) = \emptyset$.

    Second, if $r$ is a complete simulation, we also follow the steps of the previous item, constructing the set $U$ and showing that it satisfies all the needed conditions.
    In case $I(x) = \emptyset$, then $U = \emptyset$ by construction and $I(y) = \emptyset$ by assumption.
    Therefore, $\emptyset = F\pi_2(U) \lecomplete \beta(y) = \emptyset$.
  \item[`$\leready$':] This can be shown analogously to the previous item, just on a per-label basis, as the condition $I(x) = I(y)$ can be rephrased as $(\alpha(x))_a = \emptyset \iff (\beta(y))_a = \emptyset$ for each $a\in\A$.
  \qedhere
\end{description}

\section*{Proof of Lemma~\ref{lem:modal-ts-conditions}}

\noindent
Suppose we have an exact square $f\cdot u = g\cdot v$.
By exactness, for each pair $(x,y)$ such that $f(x) = g(y)$ we may choose an element $p_{x,y} \in P$ such that $u(p_{x,y}) = x$ and $v(p_{x,y}) = y$.
Let $(U_1,V_1)\in FX$ and $(U_2,V_2)\in FY$, and construct $(U,V)\in FP$ as follows:
For every $(a,x)\in U_1$ and $(b,y)\in U_2$ such that $f(x) = g(y)$ and $a \vle_\A b$, add $(a,p_{x,y})$ to $U$, and similarly for $V_1$, $V_2$ and $V$.
As $U_i\subseteq V_i$ for $i=1,2$, we also have $U\subseteq V$.
It is now straightforward to verify that $(U_1,V_1) \le_{FX} Fu(U,V)$ and $Fv(U,V) \le_{FX} (U_2,V_2)$.

For coolness, let $f\colon X\to Y$ and let $(U_1,V_1)\in FX$ and $(U'_2,V'_2)\in FY$ such that $Ff(U_1,V_1) \le_{FY} (U'_2,V'_2)$.
We need to define $(U'_1,V'_1)\in FX$ such that $Ff(U'_1,V'_1) = (U'_2,V'_2)$ and $(U_1,V_1) \le_{FX} (U'_1,V'_1)$.
One possible choice is to define $V'_1 = (\A\times f)^{-1}[V'_2]$ and $U'_1 = \{(b,x)\in (\A\times f)^{-1}[U'_2] \mid \exists a\in\A.\; a\vle_\A b \land (a,x)\in U_1 \}$ and the two conditions are then easily verified.
The symmetric case can be shown analogously.

$\lambda$ is clearly monotone and preserves the unit $\vmaxel$.
We show $\V$-normality.
Let $f,g\in 2^X$ and $(U_1,V_1), (U_2,V_2)\in FX$ such that $\lambda_X(f)(U_1,V_1) = \lambda_X(g)(U_2,V_2) = \vmaxel$ and $(U_1,V_1) \le_{FX} (U_2,V_2)$.
We need to find $(U_3,V_3) \in FX$ such that $\lambda_X(f\land g)(U_3,V_3) = \vmaxel$ and $(U_1,V_1) \le_{FX} (U_3,V_3) \le_{FX} (U_2,V_2)$.
The following choice is sufficient:
\begin{gather*}
  U_3 \coloneqq \{ (a,x) \mid (a,x)\in U_1 \land\exists b\in\A.\; a\vle_\A b \land (b,x)\in U_2 \} \\
  V_3 \coloneqq \{ (a,x) \mid (a,x)\in V_1 \land\exists b\in\A.\; a\vle_\A b \land (b,x)\in V_2 \}. \qedhere
\end{gather*}

\section*{Proof of Lemma~\ref{lem:characterize-modal-ts}}

Let $\alpha\colon X\to FX$ and $\beta\colon Y\to TY$ be LSMTSs, and $r\colon\frel{X}{Y}$.

First, assume that $r$ is a $W_\lambda$-simulation, and let $x\,r\,y$.
By assumption, we have $\alpha(x)\,W_\lambda r\,\beta(y)$, which means that there exists $(U,V) \in F(X\times Y)$ such that $\alpha(x) \le_{FX} F\pi_1(U,V)$ and $F\pi_2(U,V) \le_{FY} \beta(y)$, and also $x'\,r\,y'$ for all $(a,x',y')\in V$.
Let $x\overset{a}{\dashedrightarrow} x'$.
By the first part of the assumption, there exists $b\in\A$ and $y'\in Y$ such that $a\vle_\A b$ and $(b,x',y')\in V$, and by the second part there moreover exists $c\in\A$ such that $b\vle_\A c$ and $y\overset{c}{\dashedrightarrow}y'$.
By the third part of the assumption we also have $x'\,r\,y'$ as required.
We can similarly prove that for every $y\overset{b}{\rightarrow} y'$ there exists $x\overset{a}{\rightarrow} x'$ such that $a\vle_\A b$ and $x'\,r\,y'$.

Second, assume that $r$ is a modal refinement relation.
Now we need to find $(U,V)\in F(X\times Y)$ that satisfies the properties in the first part of the proof.
The following does the trick:
\begin{gather*}
  U \coloneqq \{ (a,x',y') \mid x\overset{a}{\rightarrow}x' \land \exists b\in\A.\; a\vle_\A b \land y\overset{b}{\rightarrow}y' \} \\
  V \coloneqq \{ (a,x',y') \mid x\overset{a}{\dashedrightarrow}x' \land \exists b\in\A.\; a\vle_\A b \land y\overset{b}{\dashedrightarrow} y' \} \qedhere
\end{gather*}

\section*{Proof of Lemma~\ref{lem:metric-streams-conditions}}

\noindent
We first show that $F$ preserves exact squares.
Let $f\cdot u = g\cdot v$ be an exact square.
Let $(a,x)\in FX$ and $(b,y)\in FY$.
It suffices to find $(c,p)\in FP$ such that
\begin{equation*}
  \max(d_\A(a,c),\Delta_X(x,u(p))) + \max(d_\A(c,b),\Delta_Y(v(p),y))
  \le \max(d_\A(a,b), \Delta_Z(f(x),g(y))).
\end{equation*}
If $f(x) \neq g(y)$, then the right hand side is $\infty$, so there is nothing to show.
Otherwise, by exactness, there exists $p\in P$ such that $u(p) = x$ and $v(p) = y$.
In this case, we may put $c \coloneqq a$ and the two sides of the inequality above then both simplify to $d_\CA(a,b)$.

To see that $F$ is cool, let $f\colon X\to Y$ and let $(a,x)\in FX$ and $(b,y)\in FY$.
If $f(x) \neq y$, then $d_{FY}((a,f(x)),(b,y)) = \infty$ nothing is to show.
Otherwise, construct the pair $(b,x)$.
Then, clearly, $d_{FX}((a,x),(b,x)) = d_{FY}((a,y),(b,y))$
Proving the other inequality is completely symmetrical.

Now we show that $\lambda$ is $\V$-subadditive (monotonicity and preservation of $0$ are immediate).
Let $f,g\colon X\to[0,\infty]$ and $t_1,t_2\in FX$, where $t_1 = (a,x)$ and $t_2 = (b,y)$.
If $x\neq y$, then \eqref{eqn:v-subadd} is immediate, as its right hand side is $\infty$.
We therefore assume $x=y$ and put $t_3\coloneqq (c,x)\in FX$ in~\eqref{eqn:v-subadd}, which causes both sides of the inequality to simplify to $f(x) + g(x) + d_\A(a,b)$.

\section*{Proof of Lemma~\ref{lem:distrib-exact}}

\noindent
One straightforwardly verifies that every $FX$ is a hemimetric space, and that each $Ff$ is a nonexpansive map, so $F$ is indeed a functor.
For the proof of preservation of exact squares and coolness, we make use of the following facts, which hold over both the natural numbers and the nonnegative real numbers:
\begin{lem}\label{lem:number-fact}
  Let $a_1,\dots,a_n,b\in M$, and put $a = a_1+\dots+a_n$. Then
  \begin{enumerate}
    \item If $a \le b$, then there exist $b_1,\dots,b_n\in M$ such that $b_1+\dots + b_n = b$ and $a_i\le b_i$ for each $1 \le i \le n$.
    \item If $a \ge b$, then there exist $b_1,\dots,b_n\in M$ such that $b_1+\dots + b_n = b$ and $a_i\ge b_i$ for each $1 \le i \le n$.
    \item In both of the cases above, we have $d_M(a,b) = d_M(a_1,b_1) + \dots + d_M(a_n,b_n)$.
  \end{enumerate}
\end{lem}
\begin{proof}
All of these facts can be easily proven using a greedy strategy, starting with the assignment $b_i \coloneqq a_i$ and then correcting the terms upwards and downwards as needed.
\end{proof}

\noindent
Let \eqref{eqn:diag-pxyz} be a pullback square and let $\mu\in FX$ and $\nu\in FY$.
Define $\sigma\in FZ$ via $\sigma(z) = \min(Ff(\mu)(z), Fg(\nu)(z))$ for each $z\in Z$.
Then $\dtv(Ff(\mu),Fg(\nu)) = \dtv(Ff(\mu),\sigma) + \dtv(\sigma,Ff(\nu))$.
Using Lemma~\ref{lem:number-fact}, we can decrease some of the values of $\mu$ and $\nu$, arriving at bags/measures $\mu'$ and $\nu'$ such that $Ff(\mu') = \sigma = Fg(\nu')$ and $\dtv(Ff(\mu),\sigma) = \dtv(\mu,\mu')$ and $\dtv(\nu',\nu) = \dtv(\sigma,Ff(\nu))$.
By the assumption that we have a pullback square we can now pick an arbitrary coupling $\rho\in FP$ of $\mu'$ and $\nu'$ and the proof is complete.

For coolness, we proceed as follows.
Let $f\colon X\to Y$ be surjective, let $\mu\in FX$ and $\nu'\in FY$.
We process each $z\in Z$ individually, building a bag/measure $\mu'$ that has the same values as $\mu$, but corrected up or down to ensure that $Ff(\mu') = \nu'$ and $\dtv(\mu,\mu') = \dtv(Ff(\mu,\nu'))$
The proof for the other inequality is symmetrical.

\section*{Proof of Lemma~\ref{lem:mlmc-conditions}}

\noindent Coolness was already proven in Example~\ref{expl:running-cool}; the basic approach for preservation of exact square is not too dissimilar.
Let~\eqref{eqn:diag-pxyz} be a pullback square, let $\mu\in FX$, and let $\nu\in FY$.
Put $\epsilon = d_{FZ}(Ff(\mu),Fg(\nu))$.
If $\epsilon = \infty$, there is nothing to show, so assume $\epsilon < \infty$.
As before, we make use of the fact that an optimal coupling $\sigma\in\dfun(\A\times Z\times\A\times Z)$ of $Ff(\mu)$ and $Fg(\nu)$ exists, and because $\epsilon < \infty$ we know that $\sigma(a,z,b,z') = 0$ whenever $z\neq z'$.
We now construct a probability distribution $\rho\in FP$ such that $Fu(\rho) = \mu$ and $d_{FY}(Fv(\rho),\nu) \le \epsilon$, which will finish the proof.
We do so using the following algorithm:
\begin{enumerate}
  \item Initialize $\rho$ with zeroes.
  \item For every $(a,z,b,z)\in\A\times Z\times\A\times Z$ such that $\sigma(a,z,b,z) > 0$, every $x\in f^{-1}[z]$ and every $y\in g^{-1}[z]$, put $\delta = \mu(a,x)\cdot \frac{\nu(b,y)}{Fg(\nu)(b,z)}$ and increase $\rho(a,x,y)$ by $\delta$.
\end{enumerate}

\noindent
One can then verify that $\rho$ indeed has the claimed properties.

\section*{Proof of Lemma~\ref{lem:mlmc-representation-classical}}

\noindent
Let $\mu\in\dfun(\A\times X)$ and $\nu\in\dfun(\A\times Y)$.
Then we have, by expanding definitions,
\begin{equation*}
  \Wc_\lambda r(\mu,\nu) =
  \inf \{ \sum_{(a,x,y)\in\A\times X\times Y} \sigma(a,x,y)\cdot r(x,y) \mid \sigma\in\dfun(\A\times X\times Y), \dfun(\pi_{12})(\sigma) = \mu, \dfun(\pi_{13})(\sigma) = \nu\} 
\end{equation*}
and
\begin{multline*}
  \Kant(\Delta_\A\times r)(\mu,\nu) =
  \inf \{ \sum_{(a,x,b,y)\in \A\times X\times\A\times Y} \rho(a,x,b,y)\cdot \max(\Delta_\A(a,b), r(x,y)) \\ \mid \rho\in\dfun(\A\times X\times\A\times Y), \dfun(\pi_{12})(\rho) = \mu, \dfun(\pi_{34})(\rho) = \nu \}.
\end{multline*}
The latter infimum is taken over more terms, but we can ignore all those $\rho$ that assign positive probability to any tuple $(a,x,b,y)$ with $a\neq b$, as this blows up the entire sum to $\infty$.
Once those terms are ignored, the two infima are clearly equal.

\section*{Proof of Lemma~\ref{lem:mlmc-decompose}}

\noindent
We show the two inequalities:
\begin{itemize}
  \item ``$\le$'':
    We have $(d_\A\boxplus r) \le (d_\A\times\Delta_Y) \cdot (\Delta_\A\times r) \cdot (d_\A\times\Delta_X)$, so the claim follows because $\Kant$ is a lax extension and thus satisfies axioms (L1) and (L3).
  \item ``$\ge$'':
    Let $\mu\in\dfun(\A\times X)$ and $\nu\in\dfun(\A\times Y)$, and let $\rho\in\dfun((\A\times X)\times(\A\times Y))$ be a coupling of $\mu$ and $\nu$.
    We need to find $\sigma\in\dfun(\A\times(X\times Y))$ such that
    \begin{equation}\label{eqn:mlmc-decompose}
      \expect{\rho}(d_\A\boxplus r) \ge
        \Kant(d_\A\times\Delta_X)(\mu,F\pi_1(\sigma)) + 
        \lambda_{X\times Y}(r) +
        \Kant(d_\A\times\Delta_Y)(F\pi_2(\sigma),\nu).
    \end{equation}
    Let $f\colon ((a,x),(b,y)) \mapsto (b,(x,y))$ and put $\sigma \coloneqq Ff(\rho)$.
    Then $\mu$ and $F\pi_1(\sigma)$ have a coupling $Fg(\rho)$, where $g\colon ((a,x),(b,y)) \mapsto ((a,x),(b,x))$.
    Therefore:
    \begin{align*}
      &\Kant(d_\A\times\Delta_X)(\mu,F\pi_1(\sigma)) \\
      &\le \expect{Fg(\rho)}(d_\A\times\Delta_X) \\
      &= \sum_{((a,x),(b,y)\in(\A\times X)\times(\A\times Y))} \rho((a,x),(b,y)) \cdot d_\A(a,b)
    \end{align*}
    We also have $\lambda_{X\times Y}(r) = \sum_{((a,x),(b,y)\in(\A\times X)\times(\A\times Y))} \rho((a,x),(b,y)) \cdot r(x,y)$, and, as $F\pi_2(\sigma) = \nu$, $\Kant(d_\A\times\Delta_Y)(F\pi_2(\sigma),\nu) = 0$.
    Summing these together, we obtain~\eqref{eqn:mlmc-decompose}. \qedhere
\end{itemize}

\section*{Details for Example~\ref{expl:sim-duality}}

\noindent
We have already seen in Lemma~\ref{lem:characterize-simulations} that $W_\Box r(U,V) = \vmaxel$ iff for all $(a,x)\in U$ there exists $(a,y)\in V$ such that $x\,r\,y$.

Expanding the definition of $K_\Lambda$, we have that $K_\Diamond r(U,V) = \vmaxel$ iff for every $f\colon X\to 2$ we have that $U\vDash\Diamond_X(f)$ implies $V\vDash\Diamond_Y(r[f])$.
We now prove that $K_\Diamond r(U,V) = \vmaxel \iff W_\Box r(U,V) = \vmaxel$ by showing the two directions:
\begin{itemize}
  \item If $K_\Diamond r(U,V) = \vmaxel$, let $(a,x)\in U$ and put $f(x) = \vmaxel$ and $f(x') = \vminel$ for $x'\neq x$.
    Then $U \vDash \Diamond_X f$ and thus, by assumption, $V \vDash \Diamond_Y(r[f])$, which precisely means that there exists $(a,y)$ in $\A$ such that $x\,r\,y$.
  \item Conversely, suppose that $W_\Box r(U,V) = \vmaxel$ and let $f\colon X\to 2$ such that $U \vDash \Diamond_X(f)$.
    Then there exists $(a,x)\in U$ such that $x\vDash f$.
    By assumption we can find some $(a,y)\in V$ such that $x\,r\,y$.
    This implies that $y\vDash r[f]$ and therefore $V \vDash\Diamond_Y r[f]$.
\end{itemize}

The proofs for $K_\Box$ and $K_{\{\Box,\Diamond\}}$ are very similar.

\section*{Details for Example~\ref{expl:stream-duality}}

\noindent
Let $r\colon\frel{X}{Y}$.
We already know by Lemma~\ref{lem:characterize-metric-streams} that $W_\lambda r = d_\A\boxplus r$.
We now show that $K_\Lambda = d_\A\boxplus r$ by proving the two inequalities one by one:
\begin{itemize}
  \item Let $(a,x)\in FX$ and $(b,y)\in FY$, and define $f\colon X\to[0,\infty]$ via $f(x) \coloneqq 0$ and $f(x') \coloneqq \infty$ for $x' \neq x$.
    Then it is easily checked that $r[f](y) = r(x,y)$ and therefore
    \begin{align*}
      d_\A(a,b) + r(x,y)
      &= (d_\A(a,b) + r[f](y)) \ominus (d_\A(a,a) + f(x)) &&\by{$d_\A(a,a) = f(x) = 0$}\\
      &= \Diamonda{a}(r[f])(b,y) \ominus \Diamonda{a}(f)(a,x) &&\by{definition of $\Diamonda{a}$}\\
      &\le K_\Lambda r((a,x),(b,y)) &&\by{definition of $K_\Lambda$}
    \end{align*}
  \item Conversely, for every $c\in\A$ and $f\colon X\to[0,\infty]$,
    \begin{align*}
      \Diamonda{c}(r[f])(b,y) \ominus \Diamonda{c}(f)(a,x)
      &= (d_\A(b,c) + r[f](y)) \ominus (d_\A(a,c) + f(x))
        &&\by{definition of $\Diamonda{c}$}\\
      &\le (d_\A(b,c) \ominus d_\A(a,c)) + (r[f](y) \ominus f(x)) &&\by{fact about $\ominus$}\\
      &\le d_\A(a,b) + r(x,y). &&\by{triangle inequality and definition of $r[f]$}
    \end{align*}
\end{itemize}

\end{document}